\documentclass[pra,superscriptaddress,floatfix,11pt,letterpaper]{revtex4}
\usepackage{graphicx}
\usepackage{subfigure}
\usepackage{amssymb}
\usepackage{verbatim}
\usepackage{amsmath}
\usepackage{amscd}
\usepackage{amssymb}
\usepackage{setspace}

\newcommand{\ket}[1]{|{#1}\rangle}

\newcommand{\ncd}{\newcommand}
\ncd{\QC}{$\mbox{QC}_{\cal{C}}\;$}
\ncd{\QCpr}{${\mbox{QC}_{\cal{C}}}^\prime\;$}
\ncd{\QCns}{$\mbox{QC}_{\cal{C}}$}
\ncd{\QCprns}{${\mbox{QC}_{\cal{C}}}^\prime$}
\ncd{\cskN}{{|\phi_{\{\kappa\} } \rangle}_{{\cal{C}}_N}}
\ncd{\cskNpr}{{|\phi_{\{\kappa^\prime\} } \rangle}_{{\cal{C}}_N}}
\ncd{\cskNtil}{{|\phi_{\{\tilde{\kappa} \} } \rangle}_{{\cal{C}}_N}}
\ncd{\csk}{{|\phi_{\{\kappa\} } \rangle}_{\cal{C}}}
\ncd{\csktil}{{|\phi_{\{\tilde{\kappa} \} } \rangle}_{\cal{C}}}
\ncd{\cskf}{|\phi_{\{\kappa\} } \rangle_{\cal{C}}}
\ncd{\csktilf}{|\phi_{\{\tilde{\kappa} \} } \rangle_{\cal{C}}}
\ncd{\bracsk}{\mbox{}_{\cal{C}}\langle\phi_{\{\kappa\} }|}
\ncd{\bracsktil}{\mbox{}_{\cal{C}}\langle\phi_{\{\tilde{\kappa} \} }|}
\ncd{\nbracsk}{\mbox{}_{\cal{C}}\langle\phi_{\{\kappa\} }}
\ncd{\nbracsktil}{\mbox{}_{\cal{C}}\langle\phi_{\{\tilde{\kappa} \} }}
\ncd{\cs}{|\phi \rangle_{\cal{C}}\;}
\ncd{\csns}{|\phi \rangle_{\cal{C}}}
\ncd{\nbgh}{\text{nbgh}}
\ncd{\Sab}{S^{ab}}
\ncd{\Sba}{S^{ba}}
\ncd{\ds}{\displaystyle}
\ncd{\ovl}{\overline}

\newtheorem{fact}{Fact}
\newtheorem{definition}{Definition}

\newtheorem{lemma}{Lemma}
\newtheorem{theorem}{Theorem}
\newtheorem{corollary}{Corollary}
\newtheorem{remark}{Remark}

\newtheorem{procedure}{Procedure}

\newenvironment{proof}{{\noindent{\bf Proof:}}}{$\hfill\Box$}

\newcommand{\nc}{\newcommand}
\nc{\rnc}{\renewcommand}
\nc{\beq}{\begin{equation}}
\nc{\eeq}{{\end{equation}}}
\nc{\beqa}{\begin{eqnarray}}
\nc{\eeqa}{\end{eqnarray}}
\nc{\lbar}[1]{\overline{#1}}
\nc{\ketbra}[2]{|#1\rangle\!\langle#2|}
\nc{\braket}[2]{\langle#1|#2\rangle}
\nc{\proj}[1]{| #1\rangle\!\langle #1 |}
\nc{\avg}[1]{\langle#1\rangle}
\nc{\Rank}{\operatorname{Rank}}
\nc{\smfrac}[2]{\mbox{$\frac{#1}{#2}$}}
\nc{\Tr}{\operatorname{Tr}}
\nc{\id}{\operatorname{id}}
\nc{\1}{\openone}
\nc{\ox}{\otimes}
\nc{\dg}{\dagger}
\nc{\dn}{\downarrow}
\nc{\cA}{{\cal A}}
\nc{\cB}{{\cal B}}
\nc{\cC}{{\cal C}}
\nc{\cD}{{\cal D}}
\nc{\cE}{{\cal E}}
\nc{\cF}{{\cal F}}
\nc{\cG}{{\cal G}}
\nc{\cH}{{\cal H}}
\nc{\cI}{{\cal I}}
\nc{\cJ}{{\cal J}}
\nc{\cK}{{\cal K}}
\nc{\cL}{{\cal L}}
\nc{\cM}{{\cal M}}
\nc{\cN}{{\cal N}}
\nc{\cO}{{\cal O}}
\nc{\cP}{{\cal P}}
\nc{\cR}{{\cal R}}
\nc{\cS}{{\cal S}}
\nc{\cT}{{\cal T}}
\nc{\cX}{{\cal X}}
\nc{\cY}{{\cal Y}}
\nc{\cZ}{{\cal Z}}
\nc{\var}{\operatorname{var}}
\nc{\rar}{\rightarrow}
\nc{\lrar}{\longrightarrow}
\nc{\polylog}{\operatorname{polylog}}

\nc{\RR}{{{\mathbb R}}}
\nc{\CC}{{{\mathbb C}}}
\nc{\FF}{{{\mathbb F}}}
\nc{\NN}{{{\mathbb N}}}
\nc{\ZZ}{{{\mathbb Z}}}
\nc{\PP}{{{\mathbb P}}}
\nc{\QQ}{{{\mathbb Q}}}
\nc{\UU}{{{\mathbb U}}}
\nc{\EE}{{{\mathbb E}}}
\nc{\Icoh}{{I^{\rm coh}}}
\nc{\Qca}{{Q_{\rm ss}}}
\nc{\Qcaa}{{Q^{(1)}_{\rm ss}}}
\nc{\Dcaa}{{D^{(1)}_{{\rm ss}\rightarrow}}}
\nc{\Dca}{{D_{{\rm ss}\rightarrow}}}

\nc{\be}{\begin{equation}}
\nc{\ee}{{\end{equation}}}
\nc{\bea}{\begin{eqnarray}}
\nc{\eea}{\end{eqnarray}}
\nc{\<}{\langle}
\rnc{\>}{\rangle}
\nc{\Hom}[2]{\mbox{Hom}(\CC^{#1},\CC^{#2})}
\nc{\rU}{\mbox{U}}

\let\gconcat\sqsubset

\begin{document}

\singlespacing

\title{Graph Concatenation for Quantum Codes}

\author{Salman Beigi}
\affiliation{Institute for Quantum Information, California Institute of
Technology, Pasadena, CA 91125, USA}

\author{Isaac Chuang}
\affiliation{Department of Physics, Massachusetts Institute of
Technology, Cambridge, MA 02139, USA}
\affiliation{Department of Electric Engineering and Computer Science, Massachusetts Institute of
Technology, Cambridge, MA 02139, USA}

\author{Markus Grassl}
\affiliation{Centre for Quantum Technologies, National University of Singapore,
Singapore}

\author{Peter Shor}
\affiliation{Department of Mathematics, Massachusetts Institute of
Technology, Cambridge, MA 02139, USA}

\author{Bei Zeng}
\affiliation{Institute for Quantum Computing, University of Waterloo, Waterloo, ON N2L3G1, Canada}
\affiliation{Department of Combinatorics and Optimization, University of Waterloo, Waterloo, ON N2L3G1, Canada}

\begin{abstract}
Graphs are closely related to quantum error-correcting codes:
every stabilizer code is locally equivalent to a graph code, and
every codeword stabilized code can be described by a graph and a
classical code. For the construction of good quantum codes of
relatively large block length, concatenated quantum codes and
their generalizations play an important role. We develop a
systematic method for constructing concatenated quantum codes
based on ``graph concatenation", where graphs representing the
inner and outer codes are concatenated via a simple graph
operation called ``generalized local complementation." Our method
applies to both binary and non-binary concatenated quantum codes
as well as their generalizations.
\end{abstract}
\date{\today}
\date{February 3, 2010}

\maketitle


\section{Introduction}

The discovery of quantum error-correcting codes (QECCs) and the theory of
fault-tolerant quantum computation (FTQC) have greatly improved the long-term
prospects for quantum communication and computation technology. This general QECC-FTQC framework leads
to a remarkable threshold theorem, which indicates that noise likely poses no
fundamental barrier to the performance of large-scale quantum computations \cite{nielsenchuang}.

Stabilizer codes, a quantum analogue of classical additive codes, are
the most important class of QECCs \cite{Gottesman,GF4}. These codes
have dominated the study of QECC-FTQC for the past ten years because
of their simple construction based on Abelian groups. The recently
introduced codeword stabilized (CWS) quantum codes framework
\cite{CWS1,CWS2,CWS3} provides a unified way of constructing a larger
class of quantum codes, both stabilizer and nonadditive codes. Based
on the CWS framework, many nonadditive codes which outperform
stabilizer codes in terms of coding parameters, have been constructed.

Graphs are closely related to QECCs. It has been shown that every
stabilizer code is local Clifford equivalent to a graph code
\cite{StabGraph, GraphQuadratic}. The basic ingredients of a graph
code are a graph and a finite Abelian group from which the code can
explicitly be obtained \cite{GraphCode}. Every CWS code also has a
canonical form, where it can be fully characterized by a graph
$\mathcal{G}$ and a classical code $\mathcal{C}$ \cite{CWS1,CWS2}. So
a CWS quantum code $\mathcal{Q}$ is usually denoted by
$\mathcal{Q}=(\mathcal{G},\mathcal{C})$. When the classical code
$\mathcal{C}$ is linear, $\mathcal{Q}$ is a graph code; therefore,
graph codes, and hence stabilizer codes, are special cases of CWS
codes. FIG.~\ref{fig:codes} demonstrates the relationship between all
quantum codes, CWS codes and graph (stabilizer) codes.

\begin{figure}[hbt!]
\centering
\includegraphics[width=2in,angle=0]{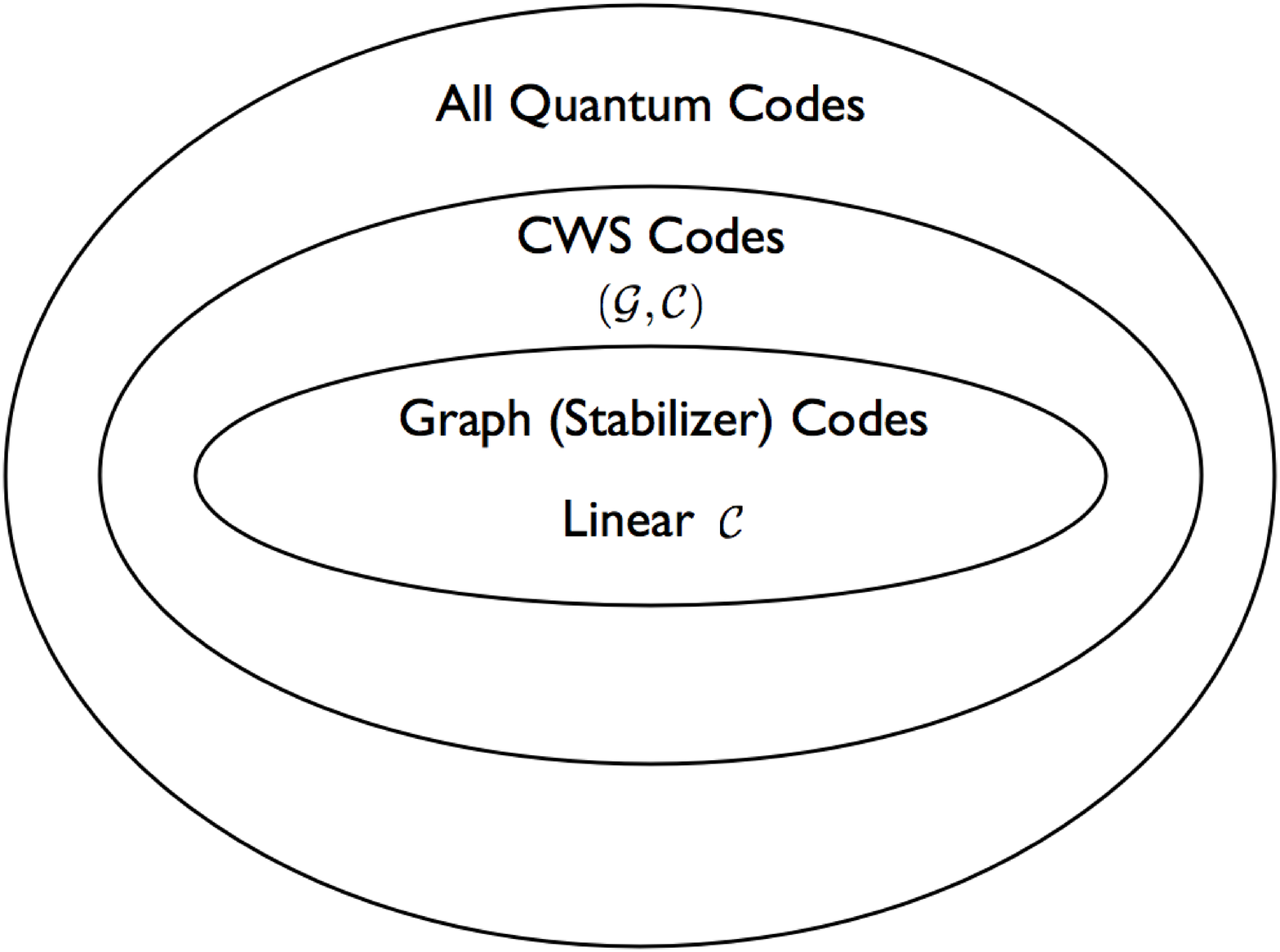}
\caption{Quantum codes}
\label{fig:codes}
\end{figure}

For the construction of good QECCs of relatively large block length
and good asymptotical performance, concatenated quantum codes and
their generalizations play an important role
\cite{nielsenchuang,Gottesman,GSSSZ,GSZ}.  Combined with the CWS
framework, families of good quantum codes, both stabilizer and
nonadditive, have been constructed \cite{GSSSZ,GSZ}. Concatenated
quantum codes also play a central role in FTQC, and the proof of the threshold
theorem \cite{nielsenchuang,Knill,Knill2,Zalka,Aharonov}.  Given the intimate relations between
graphs and quantum codes, a question that arises naturally is whether
there is a graphical description for concatenated quantum codes and
their generalizations.  Moreover, if there were such a description,
for the case where both the inner and outer codes are CWS codes, the
next question is whether the corresponding graph captures the
``quantum nature" of the concatenated code.

Previously, some related results on graph codes have been obtained.
For instance, concatenation of graph codes may be described
graphically by adding some auxiliary vertices. However, it remains
unclear what the final graph after removing those auxiliary vertices
will look like \cite{EntGraph}.  The known examples of generalized
concatenated codes only provide graphical descriptions in the case
where the outer code is of a special form \cite{GSSSZ,GSZ}.  However,
none of these previous works provides a general systematic graphical
description for constructing concatenated quantum codes.  Lack of such
a description seems to indicate that using graphs to describe quantum
codes was a very restricted approach.  This issue will be addressed in
the present work by developing a systematic method for constructing
concatenated quantum codes based on a graph operation called ``graph
concatenation."

To be more precise, we construct the concatenated quantum code
\begin{equation}
\mathcal{Q}_{c}=\mathcal{Q}_{\text{in}}\gconcat\mathcal{Q}_{\text{out}},
\end{equation}
where the inner code
$\mathcal{Q}_{\text{in}}=(\mathcal{G}_{\text{in}},\mathcal{C}_{\text{in}})$
and the outer code
$\mathcal{Q}_{\text{out}}=(\mathcal{G}_{\text{out}},\mathcal{C}_{\text{out}})$
are both CWS codes. We require $\mathcal{C}_{\text{in}}$ to be linear,
but $\mathcal{C}_{\text{out}}$ can be either linear or
nonlinear. Since $\mathcal{C}_{\text{in}}$ is linear,
$\mathcal{Q}_{\text{in}}$ is a graph (stabilizer) code.  We can then
denote the parameters of $\mathcal{Q}_{\text{in}}$ by $[[n,k,d]]_p$.
For simplicity, throughout the paper we assume that $p$ is a prime
number. When $\mathcal{Q}_{\text{in}}$ encodes $k$ qupits, the
corresponding outer code $\mathcal{Q}_{\text{out}}$ of length $n'$
must be a subspace of the Hilbert space $\mathcal{H}_{p^k}^{\otimes
  n'}$, i.e., we can denote the parameters of
$\mathcal{Q}_{\text{out}}$ by $((n',K',d'))_{p^k}$.  When
$\mathcal{C}_{\text{out}}$ is linear, then $\mathcal{Q}_{\text{out}}$
is a graph (stabilizer) code that can also be denoted by
$[[n',k',d']]_{p^k}$, where $K'=p^{kk'}$.

We now state our main result.

\noindent\textbf{Main Result:} {\it The concatenated quantum code $\mathcal{Q}_{c}$ can also be described as a
CWS code, i.e.,
\begin{equation}
\mathcal{Q}_{c}=(\mathcal{G}_{c},\mathcal{C}_{c}),
\end{equation}
and $\mathcal{Q}_{c}$ can be
constructed via the following way:
\begin{eqnarray}
\label{eq:Qc}
\mathcal{Q}_{c}&=&\mathcal{Q}_{\text{in}}\gconcat\mathcal{Q}_{\text{out}}\nonumber\\
&=&(\mathcal{G}_{\text{in}},\mathcal{C}_{\text{in}})\gconcat (\mathcal{G}_{\text{out}},\mathcal{C}_{\text{out}})\nonumber\\
&=&(\mathcal{G}_{\text{in}}\gconcat\mathcal{G}_{\text{out}},\mathcal{C}_{\text{in}}\gconcat\mathcal{C}_{\text{out}}),
\end{eqnarray}
where $\mathcal{C}_c=\mathcal{C}_{\text{in}}\gconcat\mathcal{C}_{\text{out}}$ is
the classical concatenated code with the inner code
$\mathcal{C}_{\text{in}}$ and the outer code $\mathcal{C}_{\text{out}}$, and the
graph concatenation $\mathcal{G}_{\text{in}}\gconcat\mathcal{G}_{\text{out}}$ in
Eq.~(\ref{eq:Qc}) gives the graph $\mathcal{G}_{c}$. And, we show that $\mathcal{G}_{c}$ can be obtained by
concatenating $\mathcal{G}_{\text{in}}$ and $\mathcal{G}_{\text{out}}$ via a
simple graph operation called ``generalized local complementation."}

The main advantage of constructing concatenated quantum codes via
Eq.~(\ref{eq:Qc}) is that the ``quantum part" of this construction is fully
characterized by the graph concatenation
$\mathcal{G}_{\text{in}}\gconcat\mathcal{G}_{\text{out}}$. Providing the rule for
performing this graph concatenation, the problem of constructing
concatenated quantum codes becomes purely classical, i.e.,
constructing the classical concatenated code
$\mathcal{C}_{\text{in}}\gconcat\mathcal{C}_{\text{out}}$. Despite the
restriction that $\mathcal{C}_{\text{in}}$ must be linear, our method of
graph concatenation can be applied to very general situations: both
binary and non-binary concatenated quantum codes, and
their generalizations.

This paper is organized as follows. In Sec.~\ref{sec:Sect_II}, we give a simple
example which informally demonstrates the rule of graph concatenation
via generalized local complementation. In Sec.~\ref{sec:Sect_III}, we
review definitions of graph states, CWS codes, and graph codes.  In
Sec.~\ref{sec:Sect_IV}, for a simple case that the inner code encodes only a single
qupit (i.e., $k=1$) and the outer code is also a graph code, we
provide a description of graph concatenation based on the algebraic
structure of stabilizers. We prove our main result in Sec.~\ref{sec:Sect_V}.
In Sec.~\ref{sec:Sect_VI}, we discuss the application of our main result to the situation of the
generalized concatenated quantum codes. A final discussion and conclusion
is given in Sec.~\ref{sec:Sect_VII}.

\section{A simple example and the rule}\label{sec:Sect_II}

In this section we give a simple example to demonstrate the idea
of our main result given by Eq.~(\ref{eq:Qc}). We first recall
how to describe a CWS code by a graph and a classical code; then
we demonstrate how to represent a concatenated quantum code as an
encoding graph ${\mathcal{G}}_c^{\{enc\}}$ with some auxiliary
vertices. Finally, we show how to obtain the graph
${\mathcal{G}}_c$ of the concatenated code
$\mathcal{Q}_c=(\mathcal{G}_c,\mathcal{C}_c)$, via ``generalized
local complementation" and removal of the auxiliary vertices.

\subsection{The graph and the encoding circuit of a CWS code}

Let us start by taking the outer code to be a simple $n=3$ binary CWS
quantum code
$\mathcal{Q}_{\text{out}}=(\mathcal{G}_{\text{out}},\mathcal{C}_{\text{out}})$, where
$\mathcal{G}_{\text{out}}$ is a triangle given in FIG.~\ref{fig:311}A.
$\mathcal{G}_{\text{out}}$ defines a unique quantum stabilizer state, which we denote
$\ket{\psi}_{\mathcal{G}_{\text{out}}}$.
$\mathcal{C}_{\text{out}}$ is a classical binary code of length $3$, and
can be either linear or nonlinear. A basis of the CWS code
$\mathcal{Q}_{\text{out}}$ can then be chosen as
$\{Z^{\mathbf{c}_{\text{out}}}\ket{\psi}_{\mathcal{G}_{\text{out}}}\}$, for all the codewords
$\mathbf{c}_{\text{out}}\in\mathcal{C}_{\text{out}}$ \cite{CWS1}.

If $\mathcal{C}_{\text{out}}$ is
linear, then $\mathcal{Q}_{\text{out}}$ is a graph (stabilizer) code.
The name ``graph code"  is chosen due to the fact that there is a
graphical way to represent the code $\mathcal{Q}_{\text{out}}$,
which gives both the information of $\mathcal{G}_{\text{out}}$ and
$\mathcal{C}_{\text{out}}$ \cite{GraphCode}.

\begin{figure}[hbt!]
\centering
\includegraphics[width=3.5in,angle=0]{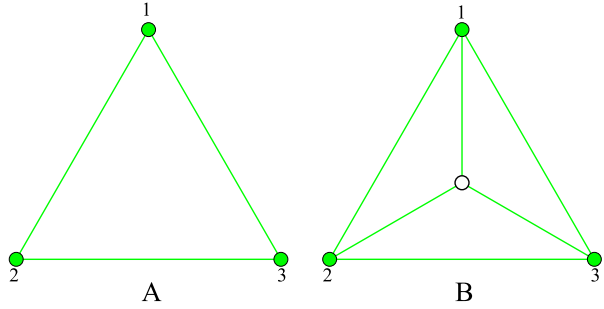}
\caption{A $[[3,1,1]]$ graph code.}
\label{fig:311}
\end{figure}

To show how to represent $\mathcal{C}_{\text{out}}$ graphically, let us
first recall the encoding circuit of
$\mathcal{Q}_{\text{out}}=(\mathcal{G}_{\text{out}},\mathcal{C}_{\text{out}})$. We use
the standard quantum circuit notations, for instance as those given
in \cite{nielsenchuang}. For a CWS code, in general the encoding
can be done by first performing a classical encoder $C_{\text{out}}^{\{enc\}}$ which
encodes the classical code $\mathcal{C}_{\text{out}}$ and then a graph
encoder $G_{\text{out}}^{\{enc\}}$ which encodes the graph state corresponding to the
graph $\mathcal{G}_{\text{out}}$ \cite{CWS1} as shown in the top left
circuit of FIG.~\ref{fig:encode}. Here $q_1,q_2,q_3$ denote
qubits $1,2,3$ in FIG.~\ref{fig:311}A, respectively.

Let us now take $\mathcal{C}_{\text{out}}=\{000,111\}$ which is linear
and gives a $[[3,1,1]]$ stabilizer code. In this case, the
classical encoder $C_{\text{out}}^{\{enc\}}$ which encodes
\begin{equation}
0\rightarrow 000,\ 1\rightarrow 111
\end{equation}
can be implemented by adding an input qubit $q_0$ and performing
controlled-NOT gates with control qubit $q_0$ and target qubits
$q_1,q_2,q_3$, followed by measuring $q_0$ in the Pauli $X$ basis
(which can be done by applying a Hadamard gate on $q_0$ and then
measuring in the Pauli $Z$ basis) as shown in the top right
circuit of FIG.~\ref{fig:encode}. Throughout the paper we always
assume that we get the desired measurement outcome (if not, we
just need to perform some local Pauli operations according to the
actual measurement outcome).

\begin{figure}[htp]
\centerline{
\begin{tabular}[t]{l@{\kern2cm}l}
\includegraphics[width=2.2in]{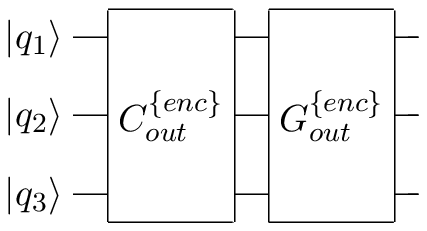}
&
\includegraphics[width=2.5in]{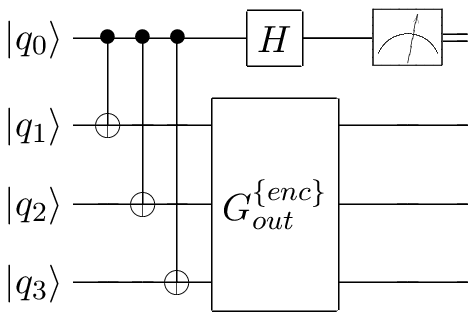}\\[5ex]
\includegraphics[height=1.5in]{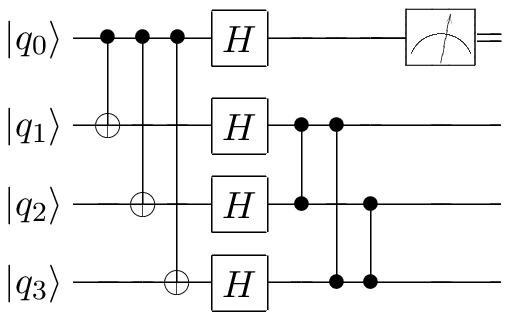}
&
\includegraphics[height=1.5in]{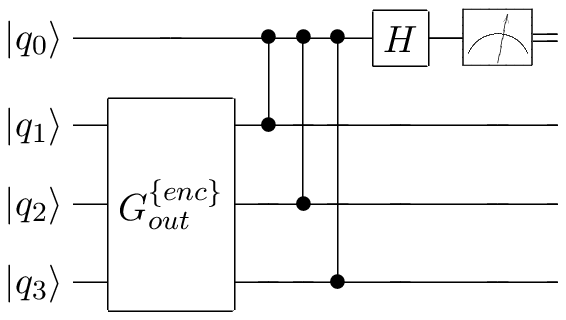}
\end{tabular}
}  
\caption{Encoding circuit for the $[[3,1,1]]$ outer code.}
\label{fig:encode}
\end{figure}

The graph encoder $G_{\text{out}}^{\{enc\}}$ consists of three Hadamard gates on
$q_1,q_2,q_3$ and three controlled-$Z$ gates between them
(controlled-$Z$ gates are applied whenever the corresponding
vertices are adjacent in graph FIG.~\ref{fig:311}A), as shown in
the bottom left circuit of FIG.~\ref{fig:encode}. Now it is clear
that we can ``move" the classical encoder $C_{\text{out}}^{\{enc\}}$ to the right of the
graph encoder $G_{\text{out}}^{\{enc\}}$ by replacing each controlled-NOT by a
controlled-$Z$, as shown in the bottom right circuit of FIG.~\ref{fig:encode}.

In the following we use the convention to modify the encoding circuit
by applying a Hadamard gate on the auxiliary qubit $q_0$ before
applying the classical encoder as shown by the left circuit of
FIG.~\ref{fig:encodecon}.  This modification can be viewed as ``a basis
change" of the input qubit $q_0$, i.e., what the ``classical encoder"
$C_{\text{out}}^{\{enc\}}$ does is then
\begin{equation}
+\rightarrow 000,\ -\rightarrow 111,
\end{equation}
where $\pm$ are the labels of the quantum states
\begin{equation}
\ket{\pm}=\frac{1}{\sqrt{2}}\left(\ket{0}\pm\ket{1}\right).
\end{equation}
This change of basis yields a non-classical encoding circuit, yet we
know that it does not make any difference for the quantum code because
by this new encoding circuit we obtain the same code space as before.
We adopt this convention throughout the paper: for any CWS code
$\mathcal{Q}=(\mathcal{G},\mathcal{C})$, we always assume that the
``classical encoder" $C^{\{enc\}}$ maps ``classical strings" in the
$\{\ket{+},\ket{-}\}$ basis to ``classical codewords" in the
$\{\ket{0},\ket{1}\}$ basis.  We will see later that this convention
naturally leads to a simple rule for graph concatenation.

Moreover, for any $[[n,k,d]]$ CWS code
$\mathcal{Q}=(\mathcal{G},\mathcal{C})$ with linear
$\mathcal{C}$ (i.e., $\mathcal{Q}$ is a graph code),
the encoding of $\mathcal{Q}$ can be applied by first performing the
graph encoder $G^{\{enc\}}$, and then the classical encoder $C^{\{enc\}}$
as follows: use $k$ input qubits; apply Hadamard on each of the
$k$ qubits; replace each controlled-NOT gate performed in the
original classical encoder $C^{\{enc\}}$ with a controlled-$Z$ gate;
finally measure each of the $k$ auxiliary qubits in the Pauli $X$
basis.

This encoding circuit can be represented graphically: add the $k$
input qubits as $k$ new vertices to the graph $\mathcal{G}$; whenever
a controlled-$Z$ is applied in the encoding circuit between an input
vertex $v$ and a vertex $v'$ of $\mathcal{G}$, add an edge between
them \cite{GraphCode}.  The corresponding graph representing the graph
code $\mathcal{Q}=(\mathcal{G},\mathcal{C})$ is denoted by
$\mathcal{G}^{\mathcal{C}}$.

Therefore, for the outer code
$\mathcal{Q}_{\text{out}}=(\mathcal{G}_{\text{out}},\mathcal{C}_{\text{out}})$
with graph $\mathcal{G}_{\text{out}}$ given in FIG.~\ref{fig:311}A,
where $\mathcal{C}_{\text{out}}=\{000,111\}$ is linear, we can insert
the input qubit $q_0$ as a new vertex (denoted by $0$) to the graph
$\mathcal{G}$ (the middle white vertex in FIG.~\ref{fig:311}B). We
then add the edges between $0$ and $1,2,3$ according to the encoding
circuit given by the bottom right circuit in FIG.~\ref{fig:encode}
(see FIG.~\ref{fig:311}B). This graph is then denoted by
$\mathcal{G}_{\text{out}}^{\mathcal{C}_{\text{out}}}$. There are two
types of vertices in
$\mathcal{G}_{\text{out}}^{\mathcal{C}_{\text{out}}}$: the input
vertices (the middle white vertex); and the output vertices (vertices
$1,2,3$).

\subsection{The encoding graph of a concatenated quantum code}
\label{sec:encodinggraph}

Now we consider the inner code
$\mathcal{Q}_{\text{in}}=(\mathcal{G}_{\text{in}},\mathcal{C}_{\text{in}})$. Notice
that due to our restriction for Eq.~(\ref{eq:Qc}),
$\mathcal{C}_{\text{in}}$ must be linear. So $\mathcal{Q}_{\text{in}}$ is a graph code
and has a graph representation $\mathcal{G}_{\text{in}}^{\mathcal{C}_{\text{in}}}$.
For simplicity we take
$\mathcal{Q}_{\text{in}}$ be a $[[2,1,1]]$ stabilizer code, which is
represented by the graph of FIG.~\ref{fig:concatenate}A on the
vertices $1, 4, 5$. (The subgraph of the vertices $4,5$
represents $\mathcal{G}_{\text{in}}$, and $1$ is the input qubit
describing the classical encoder of $\mathcal{C}_{\text{in}}$; hence,
$\mathcal{C}_{\text{in}}=\{00,11\}$.)

\begin{figure}[hbt!]
\centering
\includegraphics[width=3.5in,angle=0]{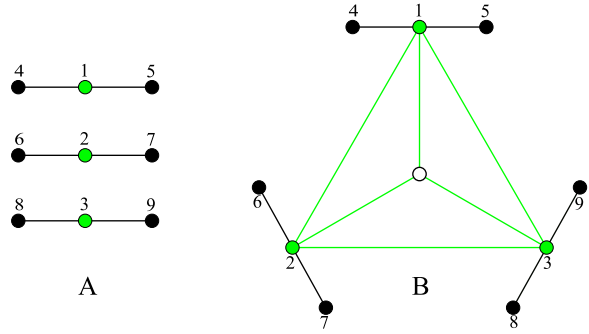}
\caption{Concatenated graph code.}
\label{fig:concatenate}
\end{figure}

To construct the concatenated code
$\mathcal{Q}_{c}=\mathcal{Q}_{\text{in}}\gconcat\mathcal{Q}_{\text{out}}$,
since the outer code has length $n'=3$, we must take three copies of
$\mathcal{G}_{\text{in}}$, to encode qubits $1,2,3$ as shown in
FIG.~\ref{fig:concatenate}A. The graphical representation of the
concatenation procedure is shown in FIG.~\ref{fig:concatenate}B.
Here, in the outer code, the middle white vertex is encoded into
vertices $1,2,3$. Then each of these vertices is encoded using the
inner code: vertex $1$ into vertices $4,5$; vertex $2$ into vertices
$6,7$; and vertex $3$ into vertices $8,9$.

We call this graphical representation of the concatenated code
$\mathcal{Q}_c$ with a linear $\mathcal{C}_{\text{out}}$
the encoding graph of $\mathcal{Q}_c$ and denote it by
${\mathcal{G}}_{\mathcal{Q}_c}^{\mathcal{C}_{\text{out}}\{enc\}}$. We
have three types of vertices in
${\mathcal{G}}_{\mathcal{Q}_c}^{\mathcal{C}_{\text{out}}\{enc\}}$: the
input vertices (the middle white vertex in our example); auxiliary
vertices which are in the subgraph $\mathcal{G}_{\text{out}}$ (vertices $1,2,3$); and
output vertices which are in the copies of $\mathcal{G}_{\text{in}}$
(vertices $4,5,6,7,8,9$). In general, if $\mathcal{C}_{\text{out}}$ is
nonlinear, similarly we can have an encoding graph of $\mathcal{Q}_c$
and denote it by ${\mathcal{G}}_{\mathcal{Q}_c}^{\{enc\}}$, which in our example is the subgraph of
FIG.~\ref{fig:concatenate}B without the middle white vertex. Therefore,
we only have two types of vertices in
${\mathcal{G}}_{\mathcal{Q}_c}^{\{enc\}}$: the auxiliary  vertices (vertices $1,2,3$);
and the output vertices(vertices $4,5,6,7,8,9$).

The encoding circuit of the concatenated code $\mathcal{Q}_{c}$ is
given by the right circuit in FIG.~\ref{fig:encodecon}, where
$G_{\text{in}}^{\{enc\}}$ denotes the graph encoder for the graph of the inner code
$\mathcal{G}_{\text{in}}$. To obtain this encoding circuit, we should
recall our convention of adding a Hadamard gate before performing
the classical encoder.

\begin{figure}[htp]
\centerline{
\vtop{\vskip0pt\hbox{\includegraphics[width=2.55in]{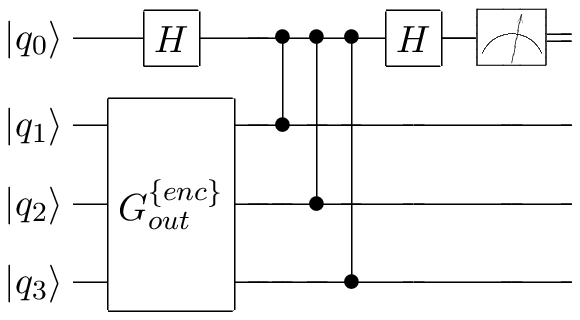}}}
\kern1cm
\vtop{\vskip0pt\hbox{\includegraphics[width=3.25in]{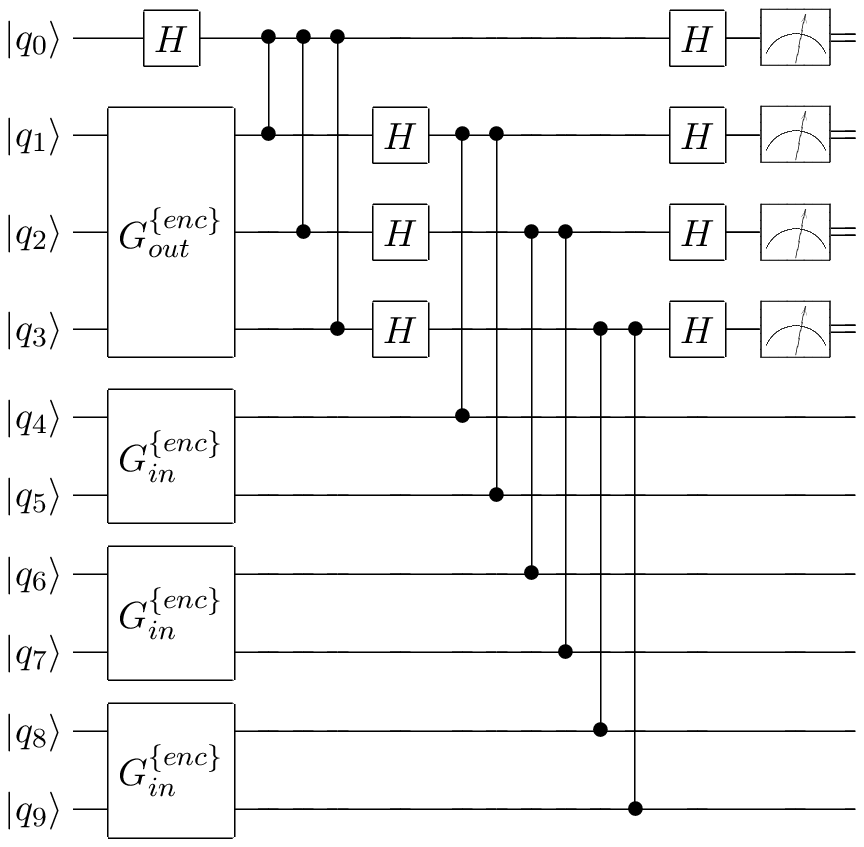}}}
}
\caption{Encoding circuit for the concatenated code with linear outer code.}
\label{fig:encodecon}
\end{figure}

In general, if $\mathcal{C}_{\text{out}}$ is nonlinear, the encoding
circuit of the concatenated code $\mathcal{Q}_{c}$ is given by the
right circuit in FIG.~\ref{fig:EnConGen}, where
${G}_{\text{in}}^{\{enc\}}$ denotes the graph encoder for the graph of
the inner code $\mathcal{G}_{\text{in}}$. Again, note that we add a
Hadamard gate before performing the classical encoder. Also, we need
to keep in mind that the ``classical encoder"
$C_{\text{out}}^{\{enc\}}$ maps ``classical strings" in the
$\{\ket{+},\ket{-}\}$ basis to ``classical codewords" in the
$\{\ket{0},\ket{1}\}$ basis.

\begin{figure}[htp]
\centerline{
\vtop{\vskip0pt\hbox{\includegraphics[width=1.65in]{EnOutGen.eps}}}
\kern1cm
\vtop{\vskip0pt\hbox{\includegraphics[width=3.0in]{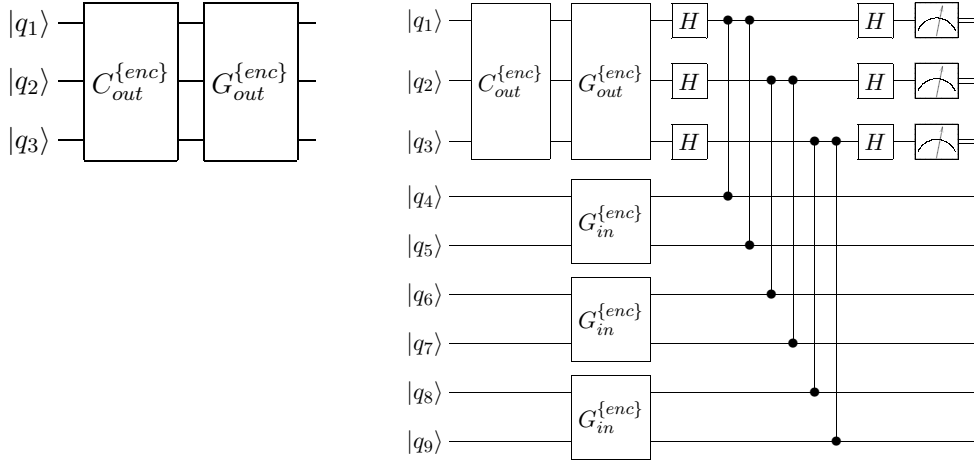}}}
}
\caption{Encoding circuit for the concatenated code with a general outer code.}
\label{fig:EnConGen}
\end{figure}

\subsection{The rule of the generalized local complementation for graph concatenation}

As shown in Sec.~\ref{sec:encodinggraph}, given a concatenated quantum
code
$\mathcal{Q}_{c}=\mathcal{Q}_{\text{in}}\gconcat\mathcal{Q}_{\text{out}}$
with a CWS outer code
$\mathcal{Q}_{\text{out}}=(\mathcal{G}_{\text{out}},\mathcal{C}_{\text{out}})$
and a graph inner code
$\mathcal{Q}_{\text{out}}=(\mathcal{G}_{\text{out}},\mathcal{C}_{\text{out}})$,
it is easy to get the encoding graph
${\mathcal{G}}_{\mathcal{Q}_c}^{\{enc\}}$. We claim (and will show
later in Sec.~\ref{sec:generalized_LC}) that the concatenated quantum
code
$\mathcal{Q}_{c}=\mathcal{Q}_{\text{in}}\gconcat\mathcal{Q}_{\text{out}}$
can also be described as a CWS code
$\mathcal{Q}_{c}=(\mathcal{G}_{c},\mathcal{C}_{c})$.  Therefore, the
real description that we want for the concatenated code
$\mathcal{Q}_{c}$ is a graph $\mathcal{G}_c$ and a classical code
$\mathcal{C}_c$ such that
$\mathcal{Q}_c=(\mathcal{G}_c,\mathcal{C}_c)$.  Also, we want the
classical code be given by the ``classical concatenation" of the
classical code of the inner and outer code, i.e.,
$\mathcal{C}_{c}=\mathcal{C}_{\text{in}}\gconcat\mathcal{C}_{\text{out}}$,
so that the quantum part can be fully taken care of by the graph
concatenation
$\mathcal{G}_{c}=\mathcal{G}_{\text{in}}\gconcat\mathcal{G}_{\text{out}}$.
Furthermore, this graph concatenation should be given by some simple
graph operations on the encoding graph
${\mathcal{G}}_{\mathcal{Q}_c}^{\{enc\}}$, i.e., we want a general
rule which gives
\begin{equation}
{\mathcal{G}}_{\mathcal{Q}_c}^{\{enc\}}\rightarrow\mathcal{G}_c,
\end{equation}
or
\begin{equation}
{\mathcal{G}}_{\mathcal{Q}_c}^{\mathcal{C}_{\text{out}}\{enc\}}\rightarrow\mathcal{G}_c^{\mathcal{C}_c},
\end{equation}
if the outer code is also a graph code. As discussed in the main
result, such a general rule does exist and we call it ``generalized
local complementation."

In this section, we demonstrate the rule of generalized local
complementation for graph concatenation by a simple example, starting
from the encoding graph given by FIG.~\ref{fig:concatenate}B. Keep in
mind that we want
\begin{equation}
\label{eq:classiccon}
\mathcal{C}_{c}=\mathcal{C}_{\text{in}}\gconcat\mathcal{C}_{\text{out}}
  =\{00,11\}\gconcat\{000,111\}=\{00\,00\,00,11\,11\,11\}.
\end{equation}

\begin{remark}
To obtain the graph $\mathcal{G}_{c}$ (or
$\mathcal{G}_{c}^{\mathcal{C}_c}$) from FIG.~\ref{fig:concatenate}B, a
naive way is to calculate the stabilizer state $\ket{\psi}$ of the
output vertices after performing Pauli $X$ measurements on all the
input and the auxiliary vertices in the encoding circuit (given by the
right graph of FIG.~\ref{fig:EnConGen} or the right graph of
FIG.~\ref{fig:encodecon}), and then to represent it as a graph state
(or a graph code).  Notice that in general it might not be possible to
represent the very code as a graph $\mathcal{G}_{c}$ (or
$\mathcal{G}_{c}^{\mathcal{C}_c}$) does not necessarily exist.
Indeed, any stabilizer state is local Clifford equivalent to a graph
state (which is not necessarily unique), any any stabilizer code is
local Clifford equivalent to a graph code, and any CWS code is local
Clifford equivalent to a standard form given by a graph and a
classical code.  However, for a general CWS code, those local Clifford
operations transform both the graph and the classical code
\cite{CWS2}.  Therefore, it is not clear that such a graph
$\mathcal{G}_c$ exists such that the classical code is obtained by
concatenation, i.e., $\mathcal{C}_{c}=\mathcal{C}_{\text{in}}\gconcat\mathcal{C}_{\text{out}}$.
\end{remark}

Now we specify the rule of generalized local complementation: given a
graph $\mathcal{G}$, for any vertex $i$, denote the set of its
adjacent vertices by $N(i)$. Also let $S$ be a subset of vertices
disjoint from $N(i)$. A generalized local complementation on $i$ with
respect to $S$ is to replace the bipartite subgraph induced on
$N(i)\cup S$ with its complement.

For an example, the generalized local complementation of the graph
shown by FIG.~\ref{fig:localcomp}A on vertex $1$ with respect to $S=\{6,7\}$
results in the graph shown by FIG.~\ref{fig:localcomp}B. Here 
$N(1)=\{2,3,4\}$. The generalized local complementation replaces the 
bipartite subgraph of vertices $\{2,3,4,6,7\}$ and edges $\{(3,6),(4,7)\}$ with
its complement (i.e. another bipartite subgraph of 
vertices $\{2,3,4,6,7\}$ and edges $\{(2,6),(4,6),(2,7),(3,7)\}$).

\begin{figure}[hbt!]
\centering
\includegraphics[width=3.5in,angle=0]{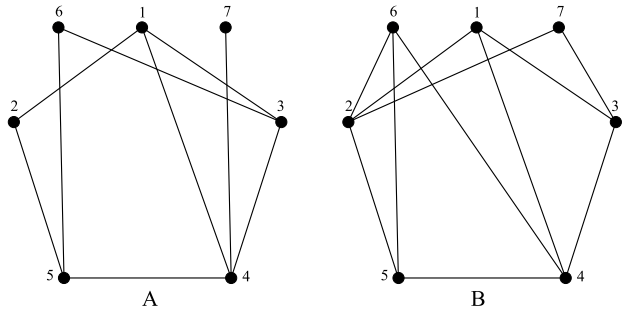}
\caption{Generalized local complementation}
\label{fig:localcomp}
\end{figure}

Now we are ready to specify the rule of graph concatenation
in terms of generalized local complementation. (Recall that our goal is
to obtain $\mathcal{G}_c$ from ${\mathcal{G}}_{\mathcal{Q}_c}^{\{enc\}}$,
or $\mathcal{G}_c^{\mathcal{C}_c}$ from
${\mathcal{G}}_{\mathcal{Q}_c}^{\mathcal{C}_{\text{out}}\{enc\}}$.)

\begin{procedure} (Graph Concatenation via Generalized Local Complementation)
\label{pro:main}

\begin{enumerate}
\item Given the graph ${\mathcal{G}}_{\mathcal{Q}_c}^{\{enc\}}$ (or
  ${\mathcal{G}}_{\mathcal{Q}_c}^{\mathcal{C}_{\text{out}}\{enc\}}$),
  for each auxiliary vertex $i$, define $S_i$ to be the set of all
  output vertices which are adjacent to $i$.
\item For each auxiliary vertex $i$, delete all the edges which connect $i$ to vertices in $S_i$.
\item For each auxiliary vertex $i$, perform generalized local
  complementation on $i$ with respect to $S_i$. Note that the order in
  which we apply those generalized local complementations does not
  matter since the whole procedure on all auxiliary vertices finally
  gives the same graph.
\item Remove all the auxiliary vertices.
\end{enumerate}
\end{procedure}

\begin{figure}[hbt!]
\centering
\includegraphics[width=4.9in,angle=0]{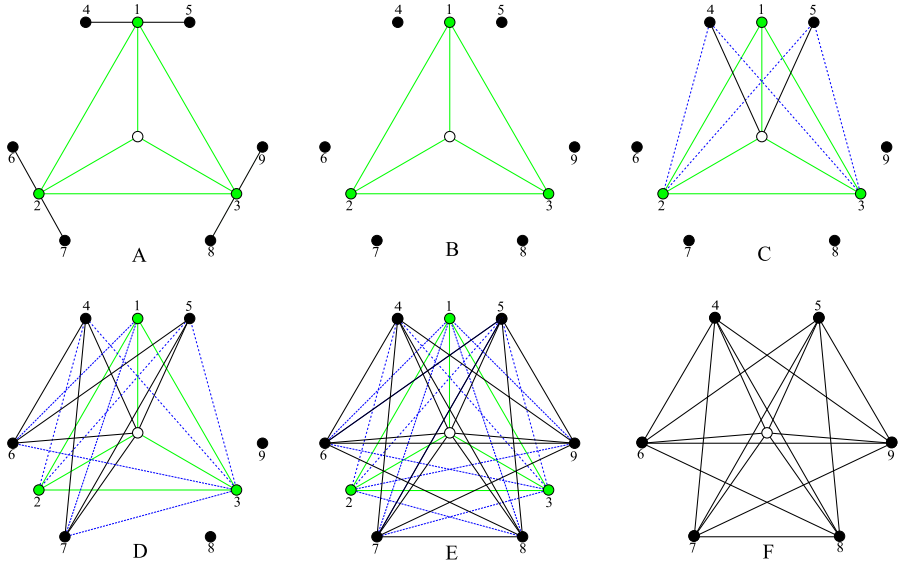}
\caption{Generalized local complementation for graph concatenation}
\label{fig:encloccomp}
\end{figure}

To demonstrate the above rules, let us start from the encoding graph
given in FIG.~\ref{fig:concatenate}B for the concatenated quantum code
with the outer code $[[3,1,1]]$ given in FIG.~\ref{fig:311}B and the
inner code $[[2,1,1]]$ given in FIG.~\ref{fig:concatenate}A. For
convenience we redraw FIG.~\ref{fig:concatenate}B in
FIG.~\ref{fig:encloccomp}A. Now from FIG.~\ref{fig:encloccomp}A we get
$S_1=\{4,5\}$, $S_2=\{6,7\}$, and $S_3=\{8,9\}$. Deleting all the
edges which connect each auxiliary vertex $i$ and vertices in $S_i$
(for $i=1,2,3$) results in FIG.~\ref{fig:encloccomp}B.  Performing
local complementation on auxiliary vertex $1$ with respect to
$S_1=\{4,5\}$ leads to FIG.~\ref{fig:encloccomp}C, where we use
dashed blue lines to show the edges that we add between output and
auxiliary vertices, and solid black lines to show other
edges. Performing local complementation in FIG.~\ref{fig:encloccomp}C
on auxiliary vertex $2$ with respect to $S_2=\{6,7\}$ leads to
FIG.~\ref{fig:encloccomp}D, and performing local complementation in
FIG.~\ref{fig:encloccomp}D on auxiliary vertex $3$ with respect to
$S_3=\{8,9\}$ leads to FIG.~\ref{fig:encloccomp}E. Removing all the
auxiliary vertices in FIG.~\ref{fig:encloccomp}E gives
FIG.~\ref{fig:encloccomp}F, which is the graph
$\mathcal{G}_c^{\mathcal{C}_c}$.

From FIG.~\ref{fig:encloccomp}F one can easily see that the rule for
concatenation of the classical codes given in
Eq.~(\ref{eq:classiccon}) holds. In general, the outer classical code
$\mathcal{C}_{\text{out}}$ is nonlinear, so we do not have the input
vertices in the encoding graph of the concatenated code. However, we
can still go through the whole procedure of the generalized local
complementations on auxiliary vertices to obtain the graph
$\mathcal{G}_{c}$. In our example this procedure is demonstrated by
subgraphs of FIG.~\ref{fig:encloccomp}A through
FIG.~\ref{fig:encloccomp}F without the middle white vertex.

\section{GRAPH STATES, CWS CODES, AND GRAPH CODES}
\label{sec:Sect_III}
In this section we review the stabilizer formalism to fix our notation
especially in the non-binary case; then we define CWS codes, graph
codes, and finally describe their encoding circuits.

\subsection{The generalized Pauli group}

Let $p$ be a prime number and $\mathbb{F}_p$ be the field of $p$
elements. A qupit is a $p$-level quantum system whose Hilbert space is
represented by the orthonormal basis $\{\vert r\rangle:\, r\in
\mathbb{F}_p\} = \{\vert 0\rangle, \vert 1\rangle, \dots , \vert
p-1\rangle \}$. Let $\omega = e^{2\pi i/p}$ be a $p$-th root of
unity. The (generalized) Pauli matrices $X$ and $Z$ are defined as
follows.

\begin{equation}
X\vert r\rangle = \vert r+1\bmod p\rangle,
\end{equation}
\begin{equation}
Z\vert r\rangle = \omega^{r}\vert r\rangle.
\end{equation}
It is clear that $X^p=Z^{p}=I$, and then we can consider the
operators $X^{a}$ and $Z^b$ where $a, b\in \mathbb{F}_p$. We have
$Z^b X^a=\omega^{ab}X^aZ^b$; therefore, $X^aZ^b$ and
$X^{a'}Z^{b'}$ commute iff $ab'-ba'=0$ (see e.g. \cite{GraphQudit} for
more details).

The group generated by the (generalized) Pauli matrices $X$ and $Z$ is
$ \{\omega^c X^{a}Z^{b}:\, a,b,c\in \mathbb{F}_p\}$ and is called the
(generalized) Pauli group. Notice that, in the binary case ($p=2$) the
Pauli group is generated by Pauli matrices $\sigma_x$ and $\sigma_z$
together with $iI$ ($i=\sqrt{-1}$).

Let $n$ be an arbitrary positive integer.  For vectors
$\mathbf{a}=(a_1, \dots, a_n)$ and $\mathbf{b}=(b_1, \dots, b_n)$ in
$\mathbb{F}_p^n$ define

\begin{equation}
X^{\mathbf{a}} = X^{a_1}\otimes \cdots \otimes X^{a_n}
\end{equation}
and
\begin{equation}
Z^{\mathbf{b}} = Z^{b_1}\otimes \cdots \otimes Z^{b_n}.
\end{equation}
Again, two Pauli matrices $X^{\mathbf{a}}Z^{\mathbf{b}}$ and
$X^{\mathbf{a'}}Z^{\mathbf{b'}}$ commute if and only if
$\mathbf{a}\mathbf{b'}-\mathbf{a'}\mathbf{b}=0$, where
$\mathbf{c}\mathbf{d}= c_1d_1+\cdots +c_nd_n$ is the usual inner
product on $\mathbb{F}_p^n$.

For simplicity, a Pauli operator $X^{\mathbf{a}}Z^{\mathbf{b}}$ is
denoted by the vector $(\mathbf{a}\, \vert\, \mathbf{b})$ of
length $2n$. Thus two Pauli operators
$X^{\mathbf{a}}Z^{\mathbf{b}}$ and
$X^{\mathbf{a'}}Z^{\mathbf{b'}}$ commute iff their corresponding
vectors are orthogonal with respect to the ``symplectic inner
product" defined by

\begin{equation}\label{eq:symp}
(\mathbf{a}\, \vert\, \mathbf{b}) * (\mathbf{a'}\, \vert\,
\mathbf{b'}) =  \mathbf{a}\mathbf{b'}-\mathbf{a'}\mathbf{b}.
\end{equation}

\subsection{Stabilizer states}

It is easy to see that for a Pauli matrix
$g=\omega^{c}X^{\mathbf{a}}Z^{\mathbf{b}}$, $g^{p}=I$. (In the case
$p=2$, the statement might only be true after replacing $G$ by $ig$ in
order to get $g^2=I$.  Having this in mind, there is no true
difference between the binary and non-binary case in the rest of the
paper.)  Therefore, the eigenvalues of $g$ are all $p$-th root of
unity. In fact, if $(\mathbf{a}\,\vert\, \mathbf{b})$ is non-zero,
then for any $i$, $\omega^i$ is an eigenvalue of $g$, and the
multiplicity of each of these $p$ eigenvalues is equal to $p^{n-1}$
\cite{GraphQudit}.

Now suppose $g_1 = \omega^{c_1}X^{\mathbf{a^1}}Z^{\mathbf{b^1}} ,
\dots ,g_k=\omega^{c_k}X^{\mathbf{a^k}}Z^{\mathbf{b^k}}$ are $k$ Pauli
matrices which pairwise commute and such that the subgroup generated
by any $k-1$ of them does not contain the other one.  Additionally, we
require that the group generated by $g_1,\dots,g_k$ does not contain a
non-trivial multiple of identity.  Since $g_1, \dots , g_k$ commute,
they can be diagonalized simultaneously.

\begin{lemma}\label{lem:eigen}
The common eigenspace of all $g_i$'s with eigenvalue $1$ is a
$p^{n-k}$-dimensional subspace.
\end{lemma}

This lemma is a well-known fact in the binary case
\cite{nielsenchuang}, and a proof for the non-binary case can be
found in \cite{GraphQudit}.

Representing the operators $g_1, \dots , g_k$ by the vectors of
length $2n$, we obtain the $k\times(2n)$ matrix
\begin{equation}
M= \begin{pmatrix}
   &  \mathbf{a^1}&  & \vline &  & \mathbf{b^1} &   \\
   &  \mathbf{a^2}&  & \vline &  & \mathbf{b^2} &   \\
   &  \vdots    &  & \vline &  & \vdots     &   \\
   &  \mathbf{a^k} & & \vline &  & \mathbf{b^k} &
\end{pmatrix},
\end{equation}
of rank $k$ (because $g_i$ is not in the subgroup generated by the
rest of $g_j$'s).  The rows of $M$ are mutually orthogonal with
respect to the symplectic inner product (see Eq.  (\ref{eq:symp})).

If we consider $n$ generators, or equivalently an $n\times (2n)$
full-rank self-orthogonal matrix $M$, Lemma~\ref{lem:eigen} implies
that the common eigenspace of all $g_i$'s with eigenvalue $1$ is a
one-dimensional subspace.  Hence there is a unique (up to a scaler)
non-zero vector $\vert \psi\rangle$ such that $g_i\vert \psi\rangle
=\vert \psi\rangle$. In fact, if we consider the group $\mathcal{S}$
generated by the $g_i$'s, for any $h\in \mathcal{S}$ we have $h\vert
\psi\rangle = \vert \psi\rangle$.  The group $\mathcal{S}$, which is a
maximal Abelian subgroup of the Pauli group modulo its center, is
called a {\it stabilizer group}, and the state $\vert \psi\rangle$ is
called a {\it stabilizer state}.

Notice that for a stabilizer state $\vert \psi\rangle$, its
stabilizer group $\mathcal{S}$ is unique; however, $\{g_1, \dots ,
g_n\}$ is just some set of generators of $\mathcal{S}$.  Suppose
$\{h_1, \dots ,h_n\}$ is another generating set of $\mathcal{S}$.
Then for any $i$ there is $u_{ij}\in \mathbb{F}_p$ such that $h_i =
g_1^{u_{i1}}\cdots g_n^{u_{\text{in}}}$.  As a result, the vector
corresponding to $h_i$ is equal to $(u_{i1}, \dots , u_{\text{in}} )
M$.

\begin{lemma}\label{lem:generator} 
Any set of generators of the stabilizer group $\mathcal{S}$ with $k$
generators can be represented by a matrix $UM$, where $U$ is an
invertible $k\times k$ matrix.
\end{lemma}

\subsection{Clifford group}

The Clifford group is the normalizer of the Pauli group.  In the
binary case, it is well-known that the Clifford group is generated by
the Hadamard gate, the phase gate, and the controlled-NOT gate
\cite{Clifford}. A characterization of the Clifford group in the
non-binary case can be found in \cite{CliffordQudit}.

Clifford operators are important in the stabilizer formalism because
they send any stabilizer state to a stabilizer state.  Suppose $\vert
\psi \rangle$ is a stabilizer state with the stabilizer group
$\mathcal{S}$. Also, let $L$ be a Clifford operator. For any $g\in
\mathcal{S}$ we have $(LgL^{\dagger})L\vert \psi\rangle = L\vert
\psi\rangle$. On the other hand, $LgL^{\dagger}$ is in
$L\mathcal{S}L^{\dagger}$ which is a subgroup of the Pauli group since
$L$ is a Clifford operator. In fact, $L\mathcal{S}L^{\dagger}$ is a
maximal Abelian subgroup of the Pauli group whose corresponding
stabilizer state is $L\vert \psi\rangle$.  Therefore, Clifford
operators send stabilizer states to stabilizer states.

Based on the characterization of the Clifford group
\cite{Clifford,CliffordQudit}, for any two stabilizer states $\vert
\psi\rangle$ and $\vert \psi'\rangle$ there is a Clifford operator $L$
such that $L\vert \psi\rangle = \vert \psi'\rangle$.  However, it does
not mean that all the stabilizer states are the same in the point of
view of quantum coding theory since the operator $L$ may completely
change the entanglement of a state.  But if we assume that
$L=L_1\otimes \cdots \otimes L_n$ is a local Clifford operator ($L$ is
the tensor product of $n$ one-qupit Clifford operators), then the
entanglement of $\vert \psi\rangle$ and $L\vert \psi\rangle$ are the
same. Based on this idea, two stabilizer states are called ``local
Clifford equivalent" if they are equivalent under the action of the
local Clifford group.

For the encoding circuits, we need only two Clifford operators which
we describe next.  Define the vector
\begin{equation}
\vert \widehat{r}\rangle = \frac{1}{\sqrt{p}}\sum_{s=0}^{p-1} \omega^{-rs} \vert s\rangle,
\end{equation}
for any $r \in \mathbb{F}_p$. $\vert \widehat{r}\rangle$ is an
eigenvector of $X$, i.e., $X\vert
\widehat{r}\rangle=\omega^{r}\vert \widehat{r}\rangle$, and $\{
\vert \widehat{0}\rangle, \dots, \vert \widehat{p-1}\rangle\}$ is
an orthonormal basis. Therefore, the operator
\begin{equation}
H\vert\widehat{r}\rangle = \vert r\rangle,
\end{equation}
which is called the
(generalized) Hadamard gate, is unitary. By definition
$HXH^{\dagger}=Z$. Also, it is easy to see that
$HZH^{\dagger}=X^{\dagger}$. Hence, $H$ is in the Clifford group.
Using the above relations the proof of the following lemma is
easy.

\begin{lemma}
\label{lem:hadamard}
Suppose $\vert \psi\rangle$ is a
stabilizer state whose stabilizer group is represented by the
$n\times (2n)$ matrix $M$. Thus, the matrix representation of the
stabilizer state $H_i\vert \psi\rangle$ (Hadamard gate is applied
on the $i$-th qupit) is obtained from $M$ by exchanging the
$i$-th and $(n+i)$-th columns and then multiplying the $i$-th
column by $-1$.
\end{lemma}

The next operator is a two-qupit gate which is called
controlled-$Z$ and is defined by
\begin{equation}
C_z\vert r\rangle \vert s\rangle = \vert r\rangle Z^r\vert s\rangle =
Z^{s}\vert r\rangle \vert s\rangle= \omega^{rs}\vert r\rangle
\vert s\rangle.
\end{equation}
We have
\begin{eqnarray}
C_zX\otimes I C_z^{\dagger} &=& X\otimes Z,\nonumber\\
C_zI\otimes X C_z^{\dagger} &=& Z\otimes X,\nonumber\\
C_zZ\otimes I C_z^{\dagger} &=& Z\otimes I,\nonumber\\
C_zI\otimes Z C_z^{\dagger} &=& I\otimes Z,
\end{eqnarray}
and thus by definition $C_z$ is in the Clifford group.

\begin{lemma}
\label{lem:c-z}
Suppose $\vert \psi\rangle$ is a
stabilizer state whose stabilizer group is represented by the
$n\times (2n)$ matrix $M$. Thus, the matrix representation of the
stabilizer state $C_{z}^{ij}\vert \psi\rangle$ (the controlled-$Z$
gate is applied on the $i$-th and $j$-th qupits) is obtained from
$M$ by adding column $i$ to column $n+j$, and column $j$ to
column $n+i$.
\end{lemma}

\subsection{Graph states}
In the following we consider graphs whose edges are labeled by
non-zero elements of $\mathbb{F}_p$. Considering the adjacency matrix
of a graph $\mathcal{G}$, we can represent it by a symmetric matrix
over $\mathbb{F}_p$ with zero diagonal. Suppose $G$ is such a matrix
of size $n\times n$. Then $M = ( I_n\, \vert\, G )$ is a full-rank
$n\times(2n)$ matrix, and all of its rows are mutually orthogonal with
respect to the symplectic inner product; therefore, $M$ represents a
stabilizer group which corresponds to a stabilizer state. Such a
stabilizer state is called a {\it graph state}, which we denote by
$\ket{\psi}_{\mathcal{G}}$.  It is well-known that any stabilizer
state is local Clifford equivalent to a graph state \cite{GraphQudit},
so to study the properties of stabilizer states it is sufficient to
restrict ourselves to graph states.

Graph states can be generated easily using only Hadamard and
controlled-$Z$ gates,
\begin{lemma}
\label{lem:graph-circuit}
The graph state corresponding to the graph with adjacency matrix
$G=(g_{ij})$ on $n$ vertices can be generated by the following
circuit. Prepare $n$ qupits in the state $\vert 0\rangle$, apply
$H^{\dagger}$ to every one of them, and then for any $i,j$ apply
$C_z^{g_{ij}}$ on qupits $i$ and $j$.
\end{lemma}

\begin{proof}
The initial state of the $n$ qupits is $\vert 0\rangle \cdots \vert
0\rangle$, which is a stabilizer state with the stabilizer group
$\{Z^{\mathbf{a}}:\, \mathbf{a}\in \mathbb{F}_p^n\}$.  This stabilizer
group corresponds to the matrix $M_0=( 0_n \,\vert\, I_n )$. According
to Lemma~\ref{lem:hadamard}, after applying $H^{\dagger}$ gates the
matrix $M_0$ will be changed to $M_1=(I_n \,\vert \, 0)$.  Also, by
Lemma~\ref{lem:c-z}, applying $C_z^{g_{ij}}$ on qupits $i$ and $j$
corresponds to adding columns $i$ and $j$ multiplied by $g_{ij}$ to
columns $n+j$ and $n+i$, respectively.  Since the first block of $M_1$
is identity, this operation is the same as to add $g_{ij}$ to the
entries $ij$ and $ji$ of the second block. Therefore, at the end we
obtain the matrix $M_2 = ( I_n\, \vert\, G )$.
\end{proof}

\subsection{Measurement on graph states}
\label{sec:measurement}

Suppose we have a graph state $\vert \psi\rangle_{\mathcal{G}}$ which
corresponds to the graph $\mathcal{G}$ with adjacency matrix $G$, and
we measure its (say) last qupit in the standard basis and get $\vert
0\rangle$.  We claim that the state after the measurement (without the
measured qubit) is also a graph state whose corresponding graph is
obtained from $\mathcal{G}$ by removing the last vertex. To see this
fact precisely notice that since $( I_n\, \vert\, G )$ represents the
stabilizer group of $\vert \psi\rangle_{\mathcal{G}}$, for any $i$, we
have $X^{\mathbf{e}_i} Z^{\mathbf{g}_i}\vert \psi\rangle_{\mathcal{G}}
= \vert \psi\rangle_{\mathcal{G}}$, where all coordinates of
$\mathbf{e}_i$ are $0$ except the $i$-th one which is $1$, and
$\mathbf{g}_i$ is the $i$-th row of $G$. Let
\begin{equation}
\vert \psi\rangle_{\mathcal{G}} 
= \sum_{r=0}^{p-1} \alpha_r\vert\phi_r\rangle\vert r\rangle,
\end{equation}
and for $1\leq i\leq n-1$ let $\mathbf{g}'_{i}$ and $\mathbf{e}'_i$ be
the vectors of length $n-1$ obtained from $\mathbf{g}_i$ and
$\mathbf{e}'_i$, respectively, by deleting the last coordinate. Thus
we have
\begin{equation}
\vert \psi\rangle_{\mathcal{G}} 
= X^{\mathbf{e}_i} Z^{\mathbf{g}_i} \vert \psi\rangle_{\mathcal{G}} 
= \sum_{r=0}^{p-1} \alpha_r Z^{{g}_{in}}X^{\mathbf{e}'_i}Z^{\mathbf{g}'_{i}} \vert \phi_r\rangle\vert r \rangle  
= \sum_{r=0}^{p-1} \alpha_r \left( \omega^{r {g}_{in}}  X^{\mathbf{e}'_i}Z^{\mathbf{g}'_{i}} \vert \phi_r\rangle\right)\vert r \rangle.
\end{equation}
As a result $X^{\mathbf{e}'_i}Z^{\mathbf{g}'_{i}} \vert \phi_0\rangle
=\vert \phi_0\rangle$, which means that $\vert \phi_0\rangle$ is a
stabilizer state with the stabilizer group generated by
$X^{\mathbf{e}'_i}Z^{\mathbf{g}'_{i}}$, $1\leq i\leq n-1$, and the
matrix representation of these generators is $( I_{n-1}\, \vert\, G')$
where $G'$ is the adjacency matrix of the graph obtained from
$\mathcal{G}$ by removing its last vertex.

\subsection{CWS codes and graph codes}
\label{sec:graph-code}
A CWS code $((n,K,d))_p$ is described by a graph $\mathcal{G}$ with
$n$ vertices and edges labeled by $\mathbb{F}_p$, together with a
classical code $\mathcal{C}$ which consists of $K$ vectors in
$\mathbb{F}_p^n$.  Such a code is denoted by
$\mathcal{Q}=(\mathcal{G}, \mathcal{C})$ \cite{CWS1,CWS2,CWS3}.

If the classical code $\mathcal{C}$ is linear, then $\mathcal{Q}$ is a
graph (stabilizer) code \cite{CWS1,CWS2,CWS3}. The parameters of such
a graph code $\mathcal{Q}=(\mathcal{G}, \mathcal{C})$ are
$[[n,k,d]]_p$, where the classical code $\mathcal{C}$ consists of
$K=p^k$ vectors in $\mathbb{F}_p^n$ that are indexed by the elements of
$\mathbb{F}_p^k$.  This $[[n,k,d]]_p$ graph code encodes $k$ qupits
into $n$ qupits in the following way. Suppose $\vert
\psi\rangle_{\mathcal{G}}$ is the graph state corresponding to
$\mathcal{G}$. To encode a state of the form $H^{\dagger}\otimes
\cdots \otimes H^{\dagger} \vert r_1 \dots r_k \rangle $ we first find
the classical codeword $\boldsymbol{\alpha}\in \mathcal{C}$ which is
indexed by $r_1\dots r_k$, and then encode $H^{\dagger}\otimes \cdots
\otimes H^{\dagger} \vert r_1 \dots r_k \rangle $ into
$Z^{\boldsymbol{\alpha}}\vert \psi\rangle_{\mathcal{G}}$
 \footnote{$Z^{\boldsymbol{\alpha}}\vert
\psi\rangle$ is usually considered as the encoded state of $\vert
r_1 \dots r_k \rangle$; however, these two codes are the same
under a change of basis, and since this change of basis is applied
locally, they have the same properties.}. Since $\mathcal{C}$ is a
linear code, it is a linear subspace of $\mathbb{F}_p^n$. We can then represent
$\mathcal{C}$ by $k$ basis vectors $\boldsymbol{\alpha}_1, \dots ,\boldsymbol{\alpha}_k$. In this case,
the state $H^{\dagger}\otimes \cdots \otimes H^{\dagger}
\vert r_1 \dots r_k \rangle $ is encoded into
$Z^{r_1\boldsymbol{\alpha}_1+\cdots + r_k\boldsymbol{\alpha}_k}\vert \psi\rangle_{\mathcal{G}}$.

The encoding circuit of a $[[n,k,d]]$ graph code is simple, as shown in the following procedure.
\begin{procedure} (Encoding circuit for a graph code)
\label{pro:enc}
\begin{enumerate}
\item First generate the graph state
$\vert \psi\rangle_{\mathcal{G}}$ using the circuit described in Lemma
\ref{lem:graph-circuit}.
\item For any $1\leq i\leq k$ apply $H^{\dagger}$ on $q_i$, where
  $q_1, \dots, q_k$ are the qupits that we want to encode.
\item For any $1\leq j\leq n$ apply
$C_z^{\alpha_{ij}}$ on $q_i$ and the $j$-th qupit of
$\vert\psi\rangle_{\mathcal{G}} $, where $\alpha_{ij}$ is the $j$-th coordinate
of $\boldsymbol{\alpha}_i$.
\item Apply $H$ to $q_1,\dots,q_k$.
\item Measure $q_1,\ldots,q_k$ in the computational basis.
\end{enumerate}
\end{procedure}

For example, the encoding circuit of the graph code with a triangle
graph and the classical code $\{000, 111\}$ can be found in the left
circuit of FIG.~\ref{fig:encodecon}. (Notice that in the binary case
$H^{\dagger}=H$.)

In general, for a graph code $\mathcal{Q}=(\mathcal{G},\mathcal{C})$
the encoding circuit can be represented graphically, and the corresponding graph is denoted by
$\mathcal{G}^{\mathcal{C}}$: consider the graph $\mathcal{G}$, for
any $1 \leq i \leq k$ add a vertex (input vertices), attach it to the vertices of
$\mathcal{G}$ (called the output vertices), and label the edge between this vertex and the
$j$-th vertex of $\mathcal{G}$ by $\alpha_{ij}$. For example,
FIG.~\ref{fig:311}B gives the graph $\mathcal{G}^{\mathcal{C}}$, where
$\mathcal{G}$ is a triangle and $\mathcal{C}=\{000, 111\}$.

\begin{remark}
The encoding circuit corresponding to the graph code with
graphical representation $\mathcal{G}^{\mathcal{C}}$ is related to
the circuit that generates the graph state $\ket{\psi}_{\mathcal{G}^{\mathcal{C}}}$
corresponding to the graph $\mathcal{G}^{\mathcal{C}}$. To see this, notice that
the steps 1,2, 3 in Procedure~\ref{pro:enc} indeed give such a graph encoder.
\end{remark}

To find the logical $X$ and $Z$ operators of an additive graph code we
first describe the stabilizer group of the logical $\vert
0\dots0\rangle_L$ state. Notice that 
\begin{equation}
\vert 0\dots 0\rangle =\frac{1}{\sqrt{p^{k}}} \sum_{r_1,\dots,r_k}  H^{\dagger}\otimes \cdots \otimes H^{\dagger}
\vert r_1 \dots r_k \rangle,
\end{equation}
and then
\begin{equation}
\vert 0\dots 0\rangle_L = \frac{1}{\sqrt{p^{k}}} \sum_{r_1,\dots,r_k}
Z^{r_1\boldsymbol{\alpha}_1+\cdots + r_k\boldsymbol{\alpha}_k}\vert \psi\rangle_{\mathcal{G}}.
\end{equation}
Therefore, all operators $ Z^{\boldsymbol{\alpha}_i}$ are in the stabilizer
group of $\vert 0\dots 0\rangle_L$, and the logical $Z$ operators
are described by the rows of the matrix
\begin{equation}
\begin{pmatrix}
     & \textbf{0} &  & \vline &  & \boldsymbol{\alpha}_1 &  \\
     & \vdots &  & \vline &  & \vdots &  \\
     & \textbf{0} &  & \vline &  & \boldsymbol{\alpha}_k &
\end{pmatrix}.
\end{equation}
Since the vectors $\boldsymbol{\alpha}_i$ are linearly independent, without
loss of generality (by a change of basis for the classical code
and reordering the qupits), we may assume that the first block of
the second part of this matrix is $I_k$. So we assume that the
matrix $(I_k \,\,\,  A)$, where $A$ is of size $k\times (n-k)$
describes a basis for $\mathcal{C}$, and the logical $Z$
operators are
\begin{equation}
\begin{pmatrix}
  & 0 & & \vline & I_k &  A
\end{pmatrix}.
\end{equation}

Assume that
\begin{equation}
G =\begin{pmatrix}
  G_1 & B \\
  B^T & G_2
\end{pmatrix},
\end{equation}
where $G_1$, $G_2$, and $B$ are of size $k\times k$,
$(n-k)\times(n-k)$, and $k\times (n-k)$, respectively. Then the
stabilizer group of the state $\vert \psi\rangle_{\mathcal{G}}$ is represented
by
\begin{equation}
\begin{pmatrix}
  I_k & 0     & \vline & G_1 & B \\
  0 & I_{n-k} & \vline & B^T & G_2
\end{pmatrix}.
\end{equation}
Now note that for any $1\leq i,j\leq k$,
$(Z^{\boldsymbol{\alpha}_i})(X^{\textbf{e}_j}Z^{\textbf{g}_j})=
\omega^{\delta_{ij}}(X^{\textbf{e}_j}Z^{\textbf{g}_j})(Z^{\boldsymbol{\alpha}_i})$, where
$\delta_{ij}$ is the Kronecker delta function. On the other hand,
the code space is invariant under $X^{\textbf{e}_j}Z^{\textbf{g}_j}$. Therefore,
the logical $X$ operators can be described by the matrix
\begin{equation}
\begin{pmatrix}
  I_k & 0     & \vline & G_1 & B
\end{pmatrix}.
\end{equation}
Also, it is not hard to see that the Pauli matrices corresponding
to the rows of
\begin{equation}
\begin{pmatrix}
  -A^T & I_{n-k} & \vline & -A^TG_1 +B^T & -A^TB+G_2
\end{pmatrix},
\end{equation}
commute with both logical $X$ and logical $Z$ operators.
Therefore, the additive graph code $\mathcal{Q}$ is described by
the stabilizer group

\begin{equation}
\label{eq:st}
\mathcal{S}=\begin{pmatrix}
  0 & 0 & \vline &I_k & A \\
  -A^T & I_{n-k} & \vline & -A^TG_1 +B^T & -A^TB+G_2
\end{pmatrix},\end{equation}
logical $Z$ operators
\begin{equation}\label{eq:l-z}\mathcal{Z}=\begin{pmatrix}
  0 & 0 & \vline & I_k & A
\end{pmatrix},
\end{equation}
and logical $X$ operators
\begin{equation}\label{eq:l-x}\mathcal{X}=\begin{pmatrix}
  I_k & 0 & \vline & G_1 & B
\end{pmatrix}.\end{equation}

\subsection{Summary of notations}

Before going into the detailed proof of the main result,
we summarize our notation.
Let $\mathcal{Q}=(\mathcal{G},\mathcal{C})$ be a CWS code. If $\mathcal{C}$
is linear, then $\mathcal{Q}$ is a graph code, where the code has a graphical
representation denoted by $\mathcal{G}^{\mathcal{C}}$. The concatenation of two CWS quantum codes
$\mathcal{Q}_{\text{in}}=(\mathcal{G}_{\text{in}},\mathcal{C}_{\text{in}})$ and
$\mathcal{Q}_{\text{out}}=(\mathcal{G}_{\text{out}},\mathcal{C}_{\text{out}})$
is denoted by $\mathcal{Q}_c=\mathcal{Q}_{\text{in}}\gconcat\mathcal{Q}_{\text{out}}$.

See Table I for the rest of notations.

\begin{table}[h]
\def\arraystretch{1.5}
\begin{tabular}{|c|l|}
  \hline
  $\mathcal{C}$ & the classical code\\ \hline
  ${C}$ & the generator matrix of the classical code $\mathcal{C}$, if $\mathcal{C}$ is linear \\ \hline
  ${C^{\{enc\}}}$ & the encoder of the classical code $\mathcal{C}$ \\ \hline
  $\mathcal{G}$ & the graph corresponding to the graph state $\ket{\psi}_{\mathcal{G}}$\\ \hline
  $G$ & the adjacency matrix of the graph $\mathcal{G}$ ($G=(g_{ij})$)\\ \hline
  $\mathcal{G}^{\mathcal{C}}$ & the graph representing the graph code $\mathcal{Q}$, if $\mathcal{C}$ is linear \\ \hline
  $G^{\{enc\}}$ & the encoding circuit of the graph $\mathcal{G}$ \\ \hline
  $\mathcal{G}_{\mathcal{Q}_c}^{\{enc\}}$ & the encoding graph of the concatenated code $\mathcal{Q}_c$\\ \hline
  $\mathcal{G}_{\mathcal{Q}_c}^{\mathcal{C}_{\text{out}}\{enc\}}$ & the encoding graph of the concatenated code $\mathcal{Q}_c$, if $\mathcal{C}_{\text{out}}$ is linear \\ \hline
\end{tabular}
\caption{Notations}
\end{table}

Most of the notations have already been given in
Sec.~\ref{sec:Sect_III}, except for
$\mathcal{G}_{\mathcal{Q}_c}^{\{enc\}}$ and
$\mathcal{G}_{\mathcal{Q}_c}^{\mathcal{C}_{\text{out}}\{enc\}}$, which
are discussed in Sec.~\ref{sec:encodinggraph} and will be explained in more details in
Sec.~\ref{sec:generalized_LC}.

\section{Concatenation of graph codes}\label{sec:Sect_IV}

In this section, we prove our main result in a simple case, where the
inner code encodes only a single qupit and the outer code is a graph
code. In this situation, we can algebraically obtain the graph and the
classical code of the concatenated code using the stabilizer
formalism. Although we will prove our main result in the general case
in Sec.~\ref{sec:generalized_LC}, we believe that the proof given in
this section is easily accessible to those who are familiar with the
stabilizer formalism.

Suppose the inner code $\mathcal{Q}_{\text{in}}=(\mathcal{G}_{\text{in}}, \mathcal{C}_{\text{in}})$
encodes only a single qupit, i.e., $\mathcal{Q}_{\text{in}}$ is an $[[n,1,d]]_p$ code.
Then from the discussion in Sec.~\ref{sec:graph-code}, it follows that
\begin{equation}\label{eq:g}{G}_{\text{in}}= \begin{pmatrix}
  0 & \mathbf{y} \\
  \mathbf{y}^T & H'
\end{pmatrix},
\end{equation}
 and
\begin{equation}
\mathcal{C}_{\text{in}}= \{ 0, (1 \,\,\,  \mathbf{b}) \},
\end{equation}
i.e.,
\begin{equation}\label{eq:b}
{C}_{\text{in}}= (1 \,\,\,  \mathbf{b}),
\end{equation}
where both $\mathbf{y}$ and $\mathbf{b}$ are vectors of length $n-1$. (Notice that
$\mathcal{Q}_{\text{in}}$ encodes one qupit; thus, ${C}_{\text{in}}$ is one-dimensional.)

Since $\mathcal{Q}_{\text{in}}$ encodes one qupit, the corresponding
outer code
$\mathcal{Q}_{\text{out}}=(\mathcal{G}_{\text{out}},\mathcal{C}_{\text{out}})$
is an $((n',K',d'))_p$ code.  In this section we assume that
$\mathcal{C}_{\text{out}}$ is linear, so $\mathcal{Q}_{\text{out}}$ is
a graph code with parameters $[[n',k',d']]_p$, where $K'=p^{k'}$. Then
from the discussion in Sec.~\ref{sec:graph-code}, we have
\begin{equation}
{G}_{\text{out}} =\begin{pmatrix}
  G_1 & B \\
  B^T & G_2
\end{pmatrix},\end{equation}
and the rows of
\begin{equation}
\label{eq:b_out}
C_{\text{out}}=(I_{k'} \,\,\,  A)
\end{equation}
form a basis for $\mathcal{C}_{\text{out}}$.

Thus by Eqs.~(\ref{eq:st})--(\ref{eq:l-x}) the
stabilizer group of $\mathcal{Q}_{\text{in}}$ is

\begin{equation} \label{eq:st-h}
\mathcal{S}_{\text{in}}=\begin{pmatrix}
  0 & {0} & \vline & 1 & \mathbf{b} \\
  -\mathbf{b}^T & I_{n-1} & \vline & \mathbf{y}^T & -\mathbf{b}^T\mathbf{y}+H'
\end{pmatrix},
\end{equation}
its logical operator $Z$ is given by
\begin{equation}\label{eq:z-h}\mathcal{Z}_{\text{in}}=\begin{pmatrix}
  0 & {0}& \vline & 1 & \mathbf{b}
\end{pmatrix},\end{equation}
and its logical operator $X$ is given by
\begin{equation}\label{eq:x-h}\mathcal{X}_{\text{in}}=\begin{pmatrix}
  1 & {0} & \vert & 0 & \mathbf{y}
\end{pmatrix}.
\end{equation}

The concatenated code
$\mathcal{Q}_c=\mathcal{Q}_{\text{in}}\gconcat\mathcal{Q}_{\text{out}}$ is a
quantum code which encodes $k'$ qupits into $nn'$ qupits as
follows; it first encodes $k'$ qupits into $n'$ qupits using
$\mathcal{Q}_{\text{out}}$, and then encodes any of the $n'$ qupits into
$n$ qupits based on $\mathcal{Q}_{\text{in}}$.

The main result of this section is given by the following theorem, which states that
the concatenated code
$\mathcal{Q}_c=\mathcal{Q}_{\text{in}}\gconcat\mathcal{Q}_{\text{out}}$ is also a graph code.
The corresponding adjacency matrix of the graph and
the generator matrix of the classical code can be 
computed directly from the adjacency matrices and the generator
matrices of the inner and outer codes.

\begin{theorem}
\label{thm:main-1}
Suppose $\mathcal{Q}_{\text{out}}=(\mathcal{G}_{\text{out}},
\mathcal{C}_{\text{out}})$ and
$\mathcal{Q}_{\text{in}}=(\mathcal{G}_{\text{in}},
\mathcal{C}_{\text{in}})$ are $[[n', k', d']]_p$ and $[[n, k, d]]_p$
graph codes, respectively, (where $k=1$) as described by
Eqs.~(\ref{eq:g})--(\ref{eq:b_out}). Then the concatenated code
$\mathcal{Q}_c=\mathcal{Q}_{\text{in}}\gconcat\mathcal{Q}_{\text{out}}=(\mathcal{G}_c,\mathcal{C}_c)$
is a graph code described by the graph $\mathcal{G}_c$ with adjacency
matrix
\begin{equation}
\label{eq:con-g}
G_c=G_{\text{in}}\otimes I_{n'} + \begin{pmatrix}
  1 \\
  \mathbf{b}^T
\end{pmatrix}(1 \,\,\,  \mathbf{b}) \otimes
G_{\text{out}},
\end{equation}
and the classical code with generator matrix
\begin{equation}
\label{eq:b-a}
{C}_c=(1 \,\,\, \mathbf{b})  \otimes (I_{k'} \,\,\,
A),
\end{equation}
i.e., the classical code is obtained by concatenation as well:
\begin{equation}
\label{eq:classcon}
\mathcal{C}_c=\mathcal{C}_{\text{in}}\gconcat\mathcal{C}_{\text{out}}.
\end{equation}
\end{theorem}

\begin{proof}
Let us first show that a basis of $\mathcal{C}_c$ is described by
Eq.~(\ref{eq:b-a}). To find $\mathcal{Z}_{c}$, the logical $Z$
operators of $\mathcal{Q}_c$, we should first consider the logical $Z$
operators of $\mathcal{Q}_{\text{out}}$, and then replace any Pauli
matrix $X$ and $Z$ of those operators with the logical $X$ and $Z$
operators of $\mathcal{Q}_{\text{in}}$. For example, if the logical
$Z$ operator acting on the first encoded qupit in
$\mathcal{Q}_{\text{out}}$ is $Z^{(1, 1, 0,\dots , 0)}$, the logical
$Z$ operator acting on the first qupit in $\mathcal{Q}_c$ is
$Z^{((1,\mathbf{b}), (1, \mathbf{b}), 0\dots 0)}$ because by
Eq.~(\ref{eq:z-h}) the logical $Z$ operator of
$\mathcal{Q}_{\text{in}}$ is $Z^{(1, \mathbf{b})}$.  Therefore, by
changing the order of qubits we can represent this Pauli matrix by the
vector $\begin{pmatrix} 0 & \vline& (1 \,\,\, \mathbf{b})\otimes
  (1,1,0,\dots ,0)\end{pmatrix}$, where the zero before the vertical
line is actually a zero vector. Now since the logical $Z$ operators of
$\mathcal{Q}_{\text{out}}$ are represented by rows of
Eq.~(\ref{eq:l-z}), we have
\begin{equation}
\label{eq:cal-z-c}
\mathcal{Z}_c =
\begin{pmatrix} 0 & \vline& (1 \,\,\, \mathbf{b})\otimes (I_{k'} \,\,\, A)
\end{pmatrix}.
\end{equation}
Equivalently, $(1 \,\,\, \mathbf{b})\otimes (I_{k'} \,\,\,
A)$ is a basis for the linear code $\mathcal{C}_c$.

Analogously, we compute for the logical $X$ operators of $\mathcal{Q}_c$:
\begin{eqnarray}
\label{eq:cal-x-c} \mathcal{X}_c & = & \begin{pmatrix} (1 \,\,\, 0)\otimes
(I_{k'} \,\,\, 0) & \vline& (0 \,\,\, \mathbf{y})\otimes (I_{k'} \,\,\, 0) + (1
\,\,\, \mathbf{b})\otimes (G_1 \,\,\, B)
\end{pmatrix}\\ & = &\begin{pmatrix}
  I_k & 0 & \vline & G_1\,\,\,  & B\,\,\,  & \mathbf{b}\otimes ( G_1 \, B) +\mathbf{y}\otimes (I_{k'} \, 0)
\end{pmatrix}.
\end{eqnarray}

It remains to compute the stabilizer group. The first $n'$ rows of
$\mathcal{S}_c$ are obtained from rows of $\mathcal{S}_{\text{out}}$
(Eq.~(\ref{eq:st})) by replacing any $X$ and $Z$ with the logical $X$
and $Z$ of the inner code. For the next $(n-1)n'$ rows note that,
$g_2, \dots, g_n$ which are the Pauli matrices corresponding to the
rows $2,\dots , n$ of $\mathcal{S}_{\text{in}}$ commute with the
logical $X$ and $Z$ of the inner code. In fact, they are in the
stabilizer group of the code space (spanned by the states $\vert
0\rangle,\dots,\vert p-1\rangle$ states).  Now since we replace any of
the $n'$ qupits of $\mathcal{Q}_{\text{out}}$ with a state in the code
space of $\mathcal{Q}_{\text{in}}$, each block of $n$ qupits in
$\mathcal{Q}_c$ should be stabilized by $g_2, \dots ,g_n$. As a
result, $\mathcal{S}_c = (\, M \,\vline\, N \, )$, where

\begin{eqnarray}\label{eq:cal-m} M =  \begin{pmatrix}
  (1\,\,\, 0) \otimes \begin{pmatrix}
    0 & 0 \\
    -A^T & I_{n'-k'} \
  \end{pmatrix} \\
  (-\mathbf{b}^T \,\,\, I_{n-1}) \otimes I_{n'}
\end{pmatrix} =  \begin{pmatrix}
  0    & 0       & 0 \\
  -A^T & I_{n'-k'} & 0 \\
  -\mathbf{b}^T\otimes \begin{pmatrix}
    I_{k'} \\
    0 \
  \end{pmatrix} & -\mathbf{b}^T\otimes\begin{pmatrix}
    0 \\
    I_{n'-k'} \
  \end{pmatrix} & I_{(n-1)n'}
\end{pmatrix},
\end{eqnarray}
and
\begin{equation}\label{eq:cal-n}
N  =  \begin{pmatrix}
  (0 \,\,\, \mathbf{y})\otimes \begin{pmatrix}
    0 & 0 \\
    -A^T & I_{n'-k'} \
  \end{pmatrix} + (1 \,\,\, \mathbf{b})\otimes \begin{pmatrix}
    I_{k'} & A \\
    -A^TG_1+B^T & -A^TB+G_2 \
  \end{pmatrix} \\
  (\mathbf{y}^T \,\,\, -\mathbf{b}^T\mathbf{y}+H')\otimes I_{n'}
\end{pmatrix}.
\end{equation}

Now to complete the proof of Theorem~\ref{thm:main-1} it is sufficient
to show that the stabilizer group, and the logical $X$ and $Z$
operators of the graph code described by Eqs.~(\ref{eq:con-g}) and
(\ref{eq:b-a}) are given by
Eqs.~(\ref{eq:cal-z-c})--(\ref{eq:cal-n}). We compute these matrices
using the construction given by Eqs.~(\ref{eq:st})--(\ref{eq:l-x}).

First of all, the classical part of the code is given by
$(I_{k'}\,\,\, B)$, where $B= (A \,\,\,  \mathbf{b} \otimes (I_{k'} \,\,\, A))$;
therefore, the logical $Z$ operators of the code are the same as
Eq.~(\ref{eq:cal-z-c}).

The block from of matrix $G_c$ of Eq.~(\ref{eq:con-g}) is
given by
\begin{equation}
G_c = \begin{pmatrix}
  K_1 & W \\
  W^T & K_{2}
\end{pmatrix},
\end{equation}
where $K_1=G_1$, $W=( B \,\,\, \mathbf{y}\otimes (I_{k'} \,\,\, 0 ) +\mathbf{b}\otimes
(G_1\,\,\, B) )$ and
\begin{equation}
K_2 = \begin{pmatrix}
   G_2 & \mathbf{y}\otimes (0 \,\,\, I_{n'-k'} ) +\mathbf{b}\otimes (B^T\,\,\, G_2) \\
   \mathbf{y}^T\otimes \begin{pmatrix}
    0 \\
    I_{n'-k'} \
  \end{pmatrix} +\mathbf{b}^T\otimes \begin{pmatrix}
    B \\
    G_2 \
  \end{pmatrix} & H'\otimes I_{n'}+\mathbf{b}^T\mathbf{b}\otimes G_2
\end{pmatrix}.
\end{equation}
Hence, the logical $X$ operator of the graph code is
\begin{equation}
\begin{pmatrix}
  I_{k'} & 0 & \vert & K_1 & W
\end{pmatrix} = \begin{pmatrix}
  I_{k'} & 0 & \vert & K_1\,\,\,  & B\,\,\,  & \mathbf{y}\otimes (I_{k'} \,\,\, 0 ) +\mathbf{b}\otimes
(G_1\,\,\, B)
\end{pmatrix},
\end{equation}
which is the same as Eq.~(\ref{eq:cal-x-c}).

The stabilizer group of the graph code is given by
\begin{equation}
\begin{pmatrix}
  0 & 0 & \vline &I_{k'} & B \\
  -B^T & I_{nn'-k'} & \vline & -B^TK_1 +W^T & -B^TW+K_2
\end{pmatrix}.
\end{equation}
By Lemma~\ref{lem:generator} this matrix describes the same
group as $\mathcal{S}_c = (\, M \,\vert\, N\,)$ because we have
\begin{equation}
\mathcal{S}_c = \begin{pmatrix}
  I_{k'} & 0 & 0 \\
  0 & I_{n'-k'} & 0 \\
  0 & -\mathbf{b}^T\otimes \begin{pmatrix}
    0 \\
    I_{n'-k'} \
  \end{pmatrix} & I_{(n-1)n'}
\end{pmatrix} \begin{pmatrix}
  0 & 0 & \vline &I_{k'} & B \\
  -B^T & I_{nn'-k'} & \vline & -B^TK_1 +W^T & -B^TW+K_2
\end{pmatrix}.
\end{equation}
\end{proof}

Notice that from Eq.~(\ref{eq:con-g}), the adjacency matrix $G_c$ does
not depend on the classical code of the outer code
($\mathcal{C}_{\text{out}}$), which indicates that Theorem
\ref{thm:main-1} could also be true even if $\mathcal{C}_{\text{out}}$
is nonlinear (in this case Eq.~(\ref{eq:b-a}) does no longer apply,
but Eq.~(\ref{eq:classcon}) may still hold).  As the stabilizer
formalism can no longer be used to handle this case, we need an
alternative proof technique which will be presented in the next
section.  This new technique is based on analyzing the encoding
circuit of the concatenated code. It can easily be extended to more
general cases, such as nonlinear outer codes, $k>1$, and even to
generalized concatenated quantum codes.

\section{Graph concatenation by generalized local complementation}
\label{sec:generalized_LC}\label{sec:Sect_V}

In this section we prove our main result based on analyzing the
encoding circuits of the concatenated code.  We start with an
alternative proof for Theorem~\ref{thm:main-1} in
Sec.~\ref{sec:AlternativeProof} for the simple case that the inner code
encodes only a single qupits, and the outer code is a graph code. This
proof is based on the rule of ``generalized local complementation."
Then in Sec.~\ref{sec:generalized_LC_B}, we show that the rule of
``generalized local complementation" can be directly applied to the
case that the outer code is a general CWS code, which is beyond the
result of Theorem~\ref{thm:main-1}.  In
Sec.~\ref{sec:generalized_LC_C}, we discuss the case where the inner
code encodes more than one qupit (i.e., $k>1$); we show that the rule
of ``generalized local complementation" given in
Sec.~\ref{sec:AlternativeProof} also applies directly to this case, and
hence completes the proof of the main result.

\subsection{Alterantive proof for Theorem 1}
\label{sec:AlternativeProof}

Recall our main goal: suppose we have two graph codes
$\mathcal{Q}_{\text{out}}=(\mathcal{G}_{\text{out}},
\mathcal{C}_{\text{out}})$ and
$\mathcal{Q}_{\text{in}}=(\mathcal{Q}_{\text{in}},
\mathcal{C}_{\text{in}})$ given by
Eqs.~(\ref{eq:g})--(\ref{eq:b_out}), where $\mathcal{Q}_{\text{in}}$
encodes a single qupit. Let $\mathcal{Q}_c
=\mathcal{Q}_{\text{in}}\gconcat\mathcal{Q}_{\text{out}}$ denote the
concatenation of the inner code $\mathcal{Q}_{\text{in}}$ and the
outer code $\mathcal{Q}_{\text{out}}$.  We would like to show that
$\mathcal{Q}_c = (\mathcal{Q}_c, \mathcal{C}_c)$, where
$\mathcal{Q}_c$ and $\mathcal{C}_c$ are given in Eqs.~(\ref{eq:con-g})
and (\ref{eq:b-a}), respectively.

In Sec.~\ref{sec:AlternativeProof_1}, we first specify the encoding
circuit of the concatenated code $\mathcal{Q}_c$, then we give the
graphical interpretation of this circuit and define the encoding graph
$\mathcal{G}_{\mathcal{Q}_c}^{\mathcal{C}_{\text{out}}\{enc\}}$ of the
concatenated code $\mathcal{Q}_c$.  Then in
Sec.~\ref{sec:AlternativeProof_2} we define the rule of ``generalized
local complementation" on a graph; we show how the encoding circuit of
the concatenated code $\mathcal{Q}_c$ can be interpreted as
generalized local complementation on the encoding graph, and how we
can obtain the graph code $\mathcal{G}_c^{\mathcal{C}_c}$ from the
encoding graph
$\mathcal{G}_{\mathcal{Q}_c}^{\mathcal{C}_{\text{out}}\{enc\}}$;
finally we show that $\mathcal{Q}_c$ and $\mathcal{C}_c$ are exactly
those given in Eqs.~(\ref{eq:con-g}) and (\ref{eq:b-a}), thereby
completing the proof.

\subsubsection{Encoding circuit and encoding graph for the concatenated code}\label{sec:AlternativeProof_1}

We have already discussed the encoding circuit of a concatenated
code in Sec.~\ref{sec:encodinggraph}. Here we state it more formally.

\begin{procedure}\label{pro:enccon} (Encoding circuit for $\mathcal{Q}_c$
with a graph outer code and an inner code encoding a single
qupit)

\begin{enumerate}
\item Apply the encoding circuit of $\mathcal{Q}_{\text{out}}$ that encodes
$k'$ qupits into $n'$ qupits which we call $q_1, \dots ,q_{n'}$,
as given by Procedure~\ref{pro:enc}.
\item Apply $n'$ copies of the circuit that gives the graph state
corresponding to $\mathcal{G}_{\text{in}}$.
\item Apply $H^\dagger$ on all qupits $q_1, \dots ,q_{n'}$.
\item Apply the corresponding controlled-$Z$ operators between these
  qupits and the graph states of $\mathcal{G}_{\text{in}}$.
\item Apply $H$ on $q_1, \dots ,q_{n'}$.
\item Measure $q_1, \dots ,q_{n'}$ in the computational basis.
\end{enumerate}
\end{procedure}

For an example, see the right circuit of FIG.~\ref{fig:encodecon}.

As discussed in Sec.~\ref{sec:encodinggraph}, Procedure
\ref{pro:enccon} can be represented by a graph which is denoted by
$\mathcal{G}_{\mathcal{Q}_c}^{\mathcal{C}_{\text{out}}\{enc\}}$. This
graph is constructed as follows: Step 1 corresponds to the graph
$\mathcal{G}_{\text{out}}^{\mathcal{C}_{\text{out}}}$; Step 2
corresponds to adding a copy of the graph $\mathcal{G}_{\text{in}}$
for each vertex $q_i$ of the graph $\mathcal{G}_{\text{out}}$, then we
have a graph on $n'+k'+nn'$ vertices; Steps 3, 4, 5 encode the $n'$
qupits of the outer code into $n'$ copies of the inner code, so we
just add edges and labels according to controlled-$Z$ gates that are
applied between these $n'$ qupits and the graph states of
$\mathcal{G}_{\text{in}}$.

For an example of the encoding graph
$\mathcal{G}_{\mathcal{Q}_c}^{\mathcal{C}_{\text{out}}\{enc\}}$, see
FIG.~\ref{fig:concatenate}B. (The corresponding encoding circuit is
given by the right circuit of FIG.~\ref{fig:encodecon}.)

\subsubsection{Graph concatenation via Generalized Local Complementation}\label{sec:AlternativeProof_2}

We now give a graphical interpretation of Steps 3, 4, 5 given in
Procedure~\ref{pro:enccon}. Notice that for $1\leq i\leq n'$ we
apply $H^{\dagger}$ to qupit $q_i$, then the corresponding
controlled-$Z$ operations between $q_i$ and the $i$-th copy of the
graph state $\mathcal{G}_{\text{in}}$, and finally $H$ on $q_i$. We
show that each of these $n'$ steps is equivalent to a generalized
local complementation on the graph.

\begin{definition} (Generalized Local Complementation)
Suppose $F=(f_{ij})$ is the adjacency matrix of a graph $\mathcal{F}$,
$i$ is a vertex of $\mathcal{F}$, and $\mathbf{f}_i$ is the $i$-th row
of $F$. Also, let $\mathbf{v}$ be a vector whose coordinates are
indexed by the vertices of $\mathcal{F}$ such that $\mathbf{v}$ is
zero on $i$ and its neighbors, i.e., $v_j=0$ if $j=i$ or $f_{ij}\neq
0$. Then the generalized local complementation at $(i,\mathbf{v})$ is
the operation which sends $F$ to
$F+\mathbf{v}^T\mathbf{f}_i+\mathbf{f}_i^T\mathbf{v}$.

Notice that, since $\mathbf{v}$ is zero on the neighbors of $i$, for any
$j$ and $k$ either $(\mathbf{v}^T\mathbf{f}_i)_{jk}$ or $(\mathbf{f}_i^T\mathbf{v})_{jk}$ is equal to
zero.
\end{definition}

To get an idea on why we call this operation the generalized local
complementation, let us consider the binary case. In this special case
$\mathbf{v}$ corresponds to a subset of vertices ($j$ belongs to this
set iff $\mathbf{v}_j= 1$). Then this operation is the same as to
replace the bipartite graph induced on the neighbors of $i$ and the
vertices in $\mathbf{v}$ with its complement. (For an example, see
FIG.~\ref{fig:localcomp}.)

\begin{theorem}
(Encoding circuit interpreted as generalized local complementation)
\label{th:glc} Consider a circuit which corresponds to a
graph $\mathcal{F}$ with the adjacency matrix $F$. Let $i$ be a vertex of $\mathcal{F}$ (or equivalently a qupit in
the circuit), and let $\mathbf{v}$ be a vector which is zero on $i$ and its
neighbors. Suppose we change the circuit by applying $H^{\dagger}$
on the $i$-th qupit, $C_z^{v_j}$ ($v_j$ is the $j$-th coordinate of $\mathbf{v}$) on the qupits $i$ and $j$, for any
$j$, and then $H$ on the $i$-th qupit. Then the resulting circuit
is equivalent to the graph $\mathcal{F}$ after the generalized local
complementation at $(i,\mathbf{v})$.
\end{theorem}

\begin{proof} 
For simplicity assume $i=1$, and let $\mathbf{f}_1 = (0 \,\,\,
\mathbf{s})$, where $\mathbf{f}_1$ is the first row of $F$. Also, let
$\mathbf{v}=(0\,\,\, \mathbf{v}')$, and $\mathcal{F'}$ be the graph
obtaining from $\mathcal{F}$ by removing its first vertex (and $F'$
its adjacency matrix). Then the stabilizer group corresponding to the
circuit is represented by
\begin{eqnarray}
(\, I \, \vert \, F\,) = \begin{pmatrix}
  1 & 0 & \vline & 0 & \mathbf{s} \\
  0 & I & \vline & \mathbf{s}^T & F'
\end{pmatrix}.
\end{eqnarray}
Then based on the translation of the action of the Hadamard gate
and controlled-$Z$ gate on the stabilizer group (Lemmas
\ref{lem:hadamard} and \ref{lem:c-z}), we can compute that
stabilizer group after applying those gates as follows:
\begin{eqnarray}
\begin{pmatrix}
  1 & 0 & \vline & 0 & \mathbf{s} \\
  0 & I & \vline & \mathbf{s}^T & F'
\end{pmatrix}  \stackrel{H^{\dagger}}{\longrightarrow}  \begin{pmatrix}  0 & 0 & \vline & -1 & \mathbf{s} \\
  \mathbf{s}^T & I & \vline & 0 & F'
\end{pmatrix}  \stackrel{C_z^{1, \mathbf{v}'}}{\longrightarrow}      \begin{pmatrix}  0 & 0 & \vline & -1 & \mathbf{s} \\
  \mathbf{s}^T & I & \vline & \mathbf{v}'^T & F'+\mathbf{s}^T\mathbf{v}'
\end{pmatrix} \\   \stackrel{H}{\longrightarrow}  \begin{pmatrix}  1 & 0 & \vline & 0 & \mathbf{s} \\
  -\mathbf{v}'^T & I & \vline & \mathbf{s}^T & F'+\mathbf{s}^T\mathbf{v}'
\end{pmatrix}
\end{eqnarray}
Now to relate this stabilizer group to a graph, we change the
set of generators by multiplying the above matrix by
\begin{eqnarray}
\begin{pmatrix}
  1 & 0 \\
  \mathbf{v}'^T & 1
\end{pmatrix},
\end{eqnarray}
which gives
\begin{eqnarray}
\begin{pmatrix}  1 & 0 & \vline & 0 & \mathbf{s} \\
  0 & I & \vline & \mathbf{s}^T & F'+\mathbf{s}^T\mathbf{v}'+\mathbf{v}'^T\mathbf{s}
\end{pmatrix}.
\end{eqnarray}
Hence, the adjacency matrix $F$ of the graph is changed to
$F+\mathbf{v}^T\mathbf{f}_i+\mathbf{f}_i^T\mathbf{v}$.
\end{proof}

A direct corollary of Theorem 2 is the following:
\begin{corollary}
\label{cor:mainsimple}
$\mathcal{Q}_c=(\mathcal{G}_c,\mathcal{C}_c)$,
and the graph $\mathcal{G}_c^{\mathcal{C}_c}$ can
be obtained from the encoding graph
$\mathcal{G}_{\mathcal{Q}_c}^{\mathcal{C}_{\text{out}}\{enc\}}$
via Procedure~\ref{pro:main}.
\end{corollary}

Notice that Corollary~\ref{cor:mainsimple} proves our main result in
the case that the inner code encodes a single qupit and the outer
code is linear.

Also, note that the resulting graph of Corollary~\ref{cor:mainsimple}
is consistent with the one given by Theorem~\ref{thm:main-1} since
they both compute the same graph. In other words, the adjacency matrix
of the graph $\mathcal{G}_c^{\mathcal{C}_{c}}$ constructed via
Corollary~\ref{cor:mainsimple} is given by Theorem~\ref{thm:main-1}.

\begin{theorem}
\label{thm:main-2}
The graph $\mathcal{G}_c^{\mathcal{C}_c}$ given by Corollary
\ref{cor:mainsimple}
is equal to the graph given by Eqs.~(\ref{eq:con-g}) and (\ref{eq:b-a}).

\end{theorem}

\begin{proof} Here we briefly describe a proof only for the binary case, and for
the validity of Eq.~(\ref{eq:con-g}). This proof can simply be captured for the
more general setting.

Based on Procedure~\ref{pro:main}, the graph on which we apply the
generalized local complementation operators has the following
subgraphs: $\mathcal{G}_{\text{out}}$ with auxiliary vertices $\{1,
\dots, n'\}$; and a copy of $\mathcal{G}_{\text{in}}$ with vertex set
$V_i$ for each auxiliary vertex $1\leq i\leq n'$.  Then for each $1\leq
i\leq n'$ we apply the generalized local complementation on $i$ with
respect to $S_i\subseteq V_i$ which is defined based on the classical
inner code.

\begin{fact}\label{fact:1} Eq.~(\ref{eq:con-g}) describes the unique graph on
the vertex set $\bigcup_i V_i$ with the following structure:
\begin{enumerate}
\item The induced subgraph on $V_i$, for every $i$, is isomorphic to
$\mathcal{G}_{\text{in}}$.
\item For every $i\neq j$, there is no edge between vertices in $V_i$ and
$V_j\setminus S_j$.
\item For every $i\neq j$, there is an edge between vertices $v\in S_i$ and
$w\in S_j$ iff $i$ and $j$ are connected in $\mathcal{G}_{\text{out}}$.
\end{enumerate}
\end{fact}

Clearly the graph with these properties is unique. Also, it is clear
that Eq.~(\ref{eq:con-g}) represents this unique graph.

Based on this fact, we will show that the graph resulting from
Procedure~\ref{pro:main} has the above structure.

\begin{fact}\label{fact:2}
During Procedure~\ref{pro:main} the changes on the subgraph induced on
$\bigcup V_i$ happen only among vertices $v\in S_i$ and $w\in S_j$ for
$i\neq j$. As a result, the final graph of Procedure~\ref{pro:main}
satisfies properties 1 and 2 of Fact~\ref{fact:1}.
\end{fact}

This is simply because in the generalized local complementations we
never touch vertices in $V_i\setminus S_i$. Furthermore,in each step,
$S_i$ is disjoint from $N(i)$ (the neighbors of vertex $i$), so there
is no change in the subgraph induced on the vertex set $S_i$.

\begin{fact}\label{fact:3} In Procedure~\ref{pro:main}, suppose we have applied
the generalized local complementation on vertices $1, 2, \dots, l$, for some
$1\leq l\leq n'$. Then for any choice of $v_i\in S_i$, for $1\leq i\leq l$, the
induced subgraph on vertices $\{v_1, \dots, v_l\} \cup \{l+1, \dots, n'\}$ is
isomorphic to $\mathcal{G}_{\text{out}}$.
\end{fact}

This fact can be proved by a simple induction on $l$.

Now we can prove the theorem. The resulting graph of
Procedure~\ref{pro:main} is a graph on the vertex set $\bigcup_i
V_i$. According to Fact~\ref{fact:2}, this graph satisfies properties
1 and 2 of Fact~\ref{fact:1}. Property 3 of Fact~\ref{fact:1} also
holds based on Fact~\ref{fact:3}. Therefore, by the uniqueness of the
graph described in Fact~\ref{fact:1}, we are done.
\end{proof}

We illustrate the graph obtained by generalized local complementation
for the code $[[25,1,9]]$ which can be obtained by self-concatenation
of the code $[[5,1,3]]$.  As a graph code, the code $[[5,1,3]]$ can be
described by a pentagon corresponding to the output nodes and a
central input node that is connected to all output nodes.  Using
auxiliary nodes, the concatenated code $[[25,1,9]]$ is show as the
left graph in FIG.~\ref{fig:pentagon}.
\begin{figure}[hbt]
\centerline{
\includegraphics[width=2in]{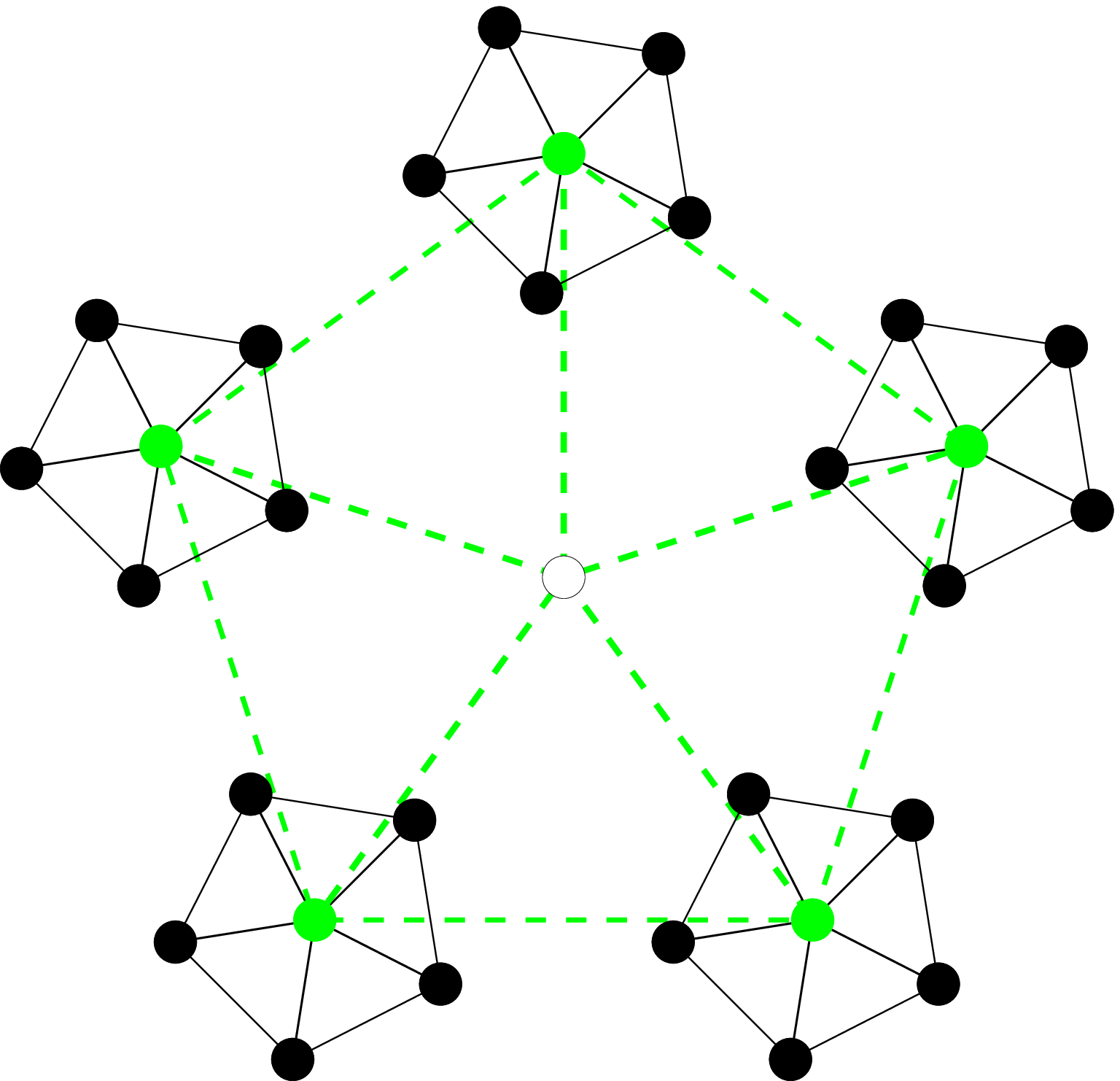}
\kern1in
\includegraphics[width=2in]{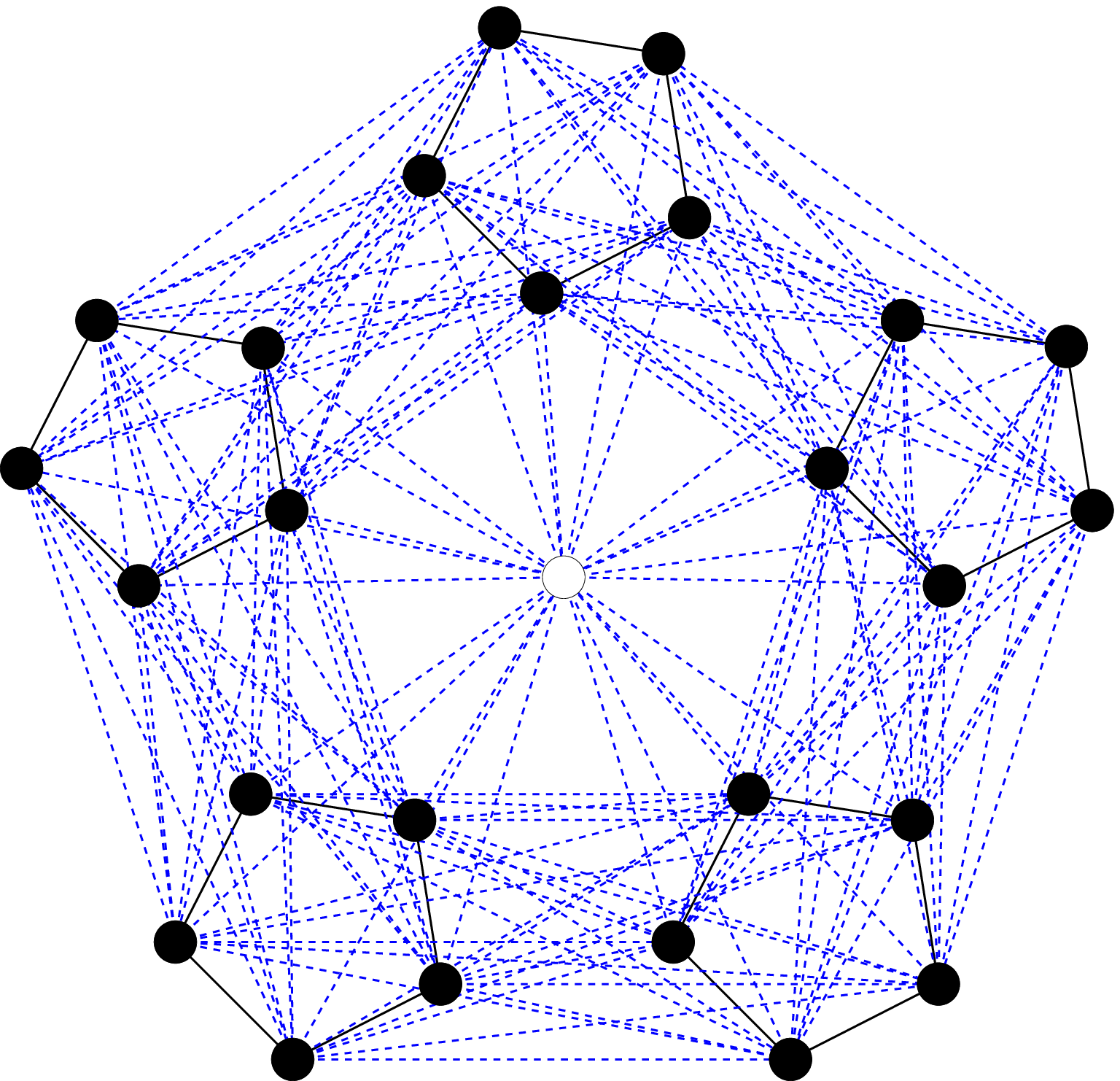}
}
\caption{Self-concatenation of the code $[[5,1,3]]$ yielding a code $[[25,1,9]]$.\label{fig:pentagon}}
\end{figure}

The outer code is given by the large pentagon with green/light dots
and dashed lines.  The five copies of the inner code correspond to the
small pentagons with black dots and solid lines. The final graph is
shown on the right of FIG.~\ref{fig:pentagon}.  The five solid black
pentagons remain, and any vertex in a small pentagon is connected with
blue/dashed lines to any vertex of the neighboring pentagons as well
as the central input node.

The situation for the self-concatenation of Steane's code $[[7,1,3]]$,
which can be realized as a cube, is shown in FIG.~\ref{fig:cube}.
\begin{figure}[hbt]
\centerline{
\includegraphics[width=2in]{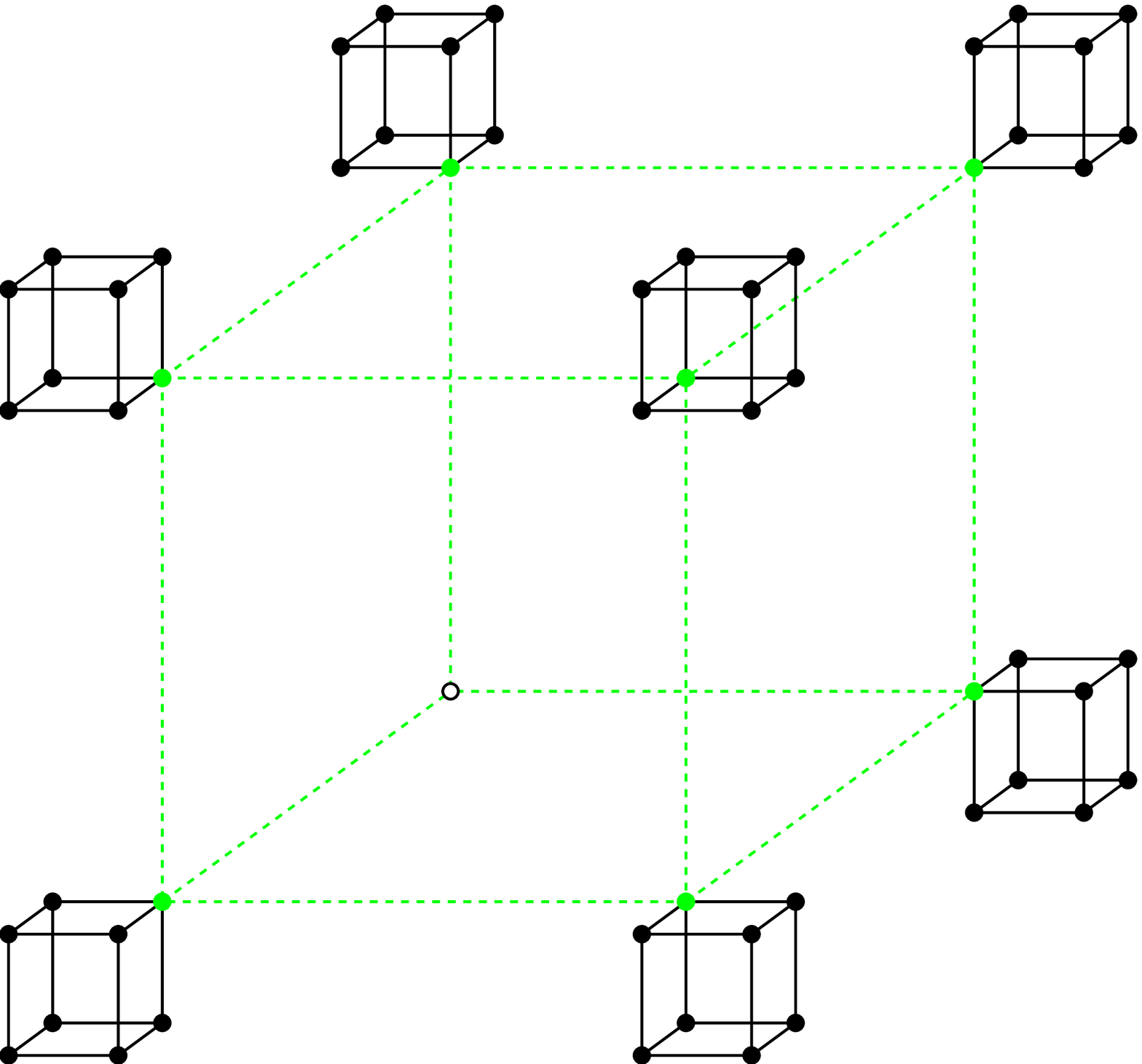}
\kern1in
\includegraphics[width=2in]{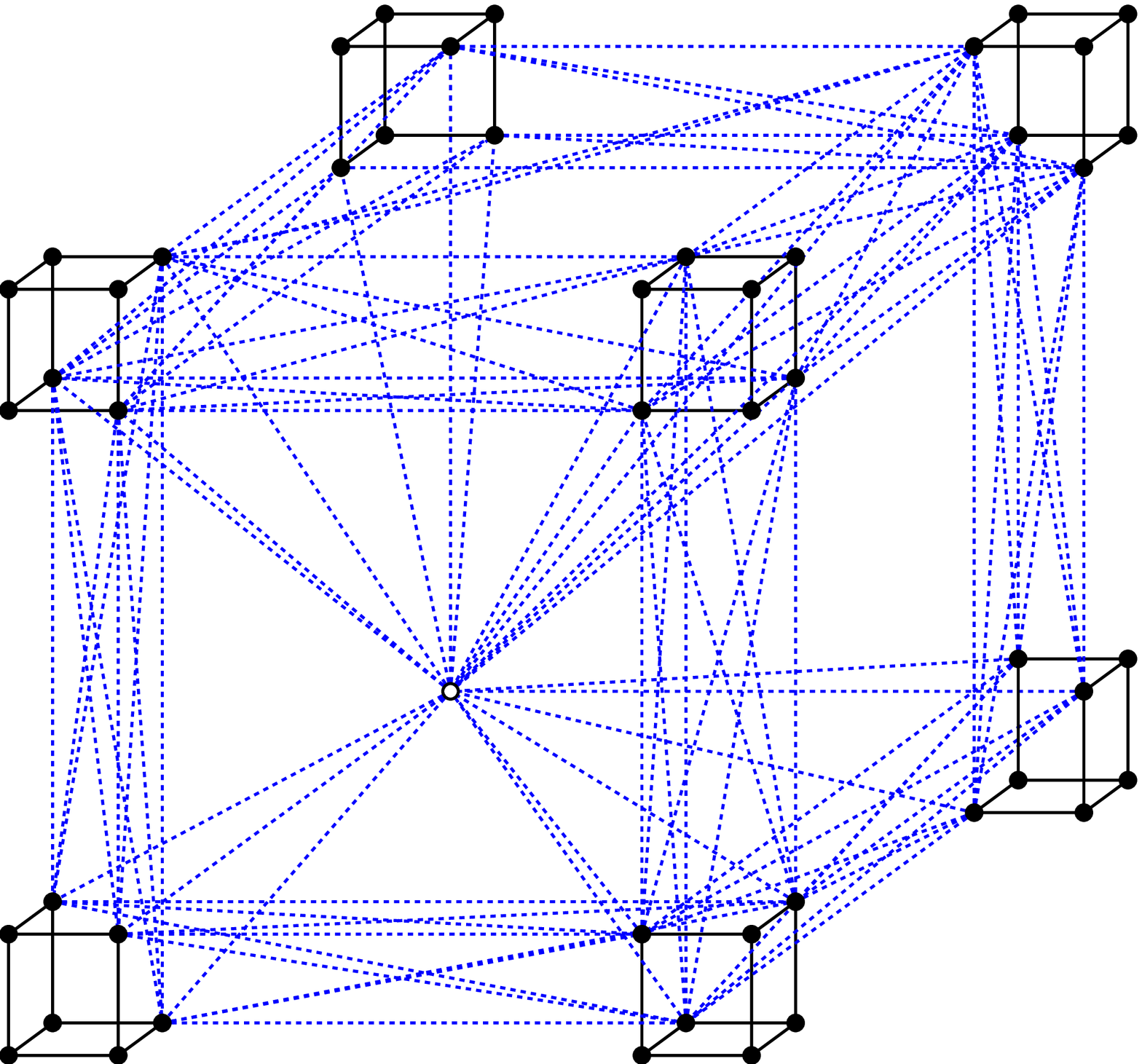}
}
\caption{Self-concatenation of Steane's code $[[7,1,3]]$ yielding a code $[[49,1,9]]$.\label{fig:cube}}
\end{figure}

\subsection{A general outer code}\label{sec:generalized_LC_B}
In this section we consider the case when the outer code is nonadditive.
The advantage of Theorem~\ref{th:glc} is that it  directly applies to this case as well.
\begin{procedure} \label{pro:genout}(Encoding circuit for $\mathcal{Q}_c$
with a general outer code and an inner code encoding a single
qupit)
\begin{enumerate}
\item Apply the encoding circuit of $\mathcal{Q}_{\text{out}}$ that
  encodes $K'$ states into $n'$ qupits which we call $q_1, \dots
  ,q_{n'}$.
\item Apply $n'$ copies of the circuit that gives the graph state
corresponding to $\mathcal{G}_{\text{in}}$.
\item Apply $H^\dagger$ on all qupits $q_1, \dots ,q_{n'}$.
\item Apply the corresponding controlled-$Z$ operators between these qupits and
the graph states of $\mathcal{G}_{\text{in}}$.
\item Apply $H$ on $q_1, \dots ,q_{n'}$.
\item Measure $q_1, \dots ,q_{n'}$ in the computational basis.
\end{enumerate}
\end{procedure}

For an example, see the right circuit of FIG.~\ref{fig:EnConGen}.

Notice that Theorem~\ref{th:glc} deals with Steps 3, 4, 5 in
Procedure~\ref{pro:genout}, which are exactly the same as Steps 3, 4,
5 as in Procedure~\ref{pro:enccon}. Therefore, whether
$\mathcal{C}_{\text{out}}$ is linear or not does not actually
matter. Consequently, Corollary~\ref{cor:mainsimple}, and thus the
main result hold even for nonlinear outer codes.

\subsection{The case $k>1$}\label{sec:generalized_LC_C}

Theorem~\ref{th:glc} can also be directly applied to the case when the inner
code encodes more than one qupit. Again, to see this we only need to specify
the encoding circuit of $\mathcal{Q}_c$.

\begin{procedure} (Encoding circuit for $\mathcal{Q}_c$
with a general outer code and an inner code encoding $k$ qupits)
\begin{enumerate}
\item Apply the encoding circuit of $\mathcal{Q}_{\text{out}}$ that encodes
$K'$ states (or $kk'$ qupits if $\mathcal{C}_{\text{out}}$ is linear)
into $kn'$ qupits which we call $q_1, \dots ,q_{kn'}$.
\item Apply $n'$ copies of the circuit that gives the graph state
corresponding to $\mathcal{G}_{\text{in}}$.
\item Apply $H^\dagger$ on all qupits $q_1, \dots ,q_{kn'}$.
\item Apply the corresponding controlled-$Z$ operators between these qupits and
the graph states of $\mathcal{G}_{\text{in}}$.
\item Apply $H$ on $q_1, \dots ,q_{kn'}$.
\item Measure $q_1, \dots ,q_{kn'}$ in the computational basis.
\end{enumerate}
\end{procedure}

Note that Steps 3, 4, 5 remain the same as those given in Procedure
\ref{pro:enccon}.  Consequently, Corollary~\ref{cor:mainsimple}, and
hence our main result hold for the case of $k>1$.

For an example, the left graph of FIG.~\ref{fig:C422} is the encoding
graph $\mathcal{G}_{\mathcal{Q}_c}^{\mathcal{C}_{\text{out}}\{enc\}}$
of the concatenated code $\mathcal{Q}_c$ with a $[[4,2,2]]_2$ inner
code, and a $[[4,2,2]]_{2^2}$ outer code.  Note that we decompose the
outer code into two copies of a qubit code $[[4,2,2]]_2$.  Hence there
are $kn'=2\times 4=8$ auxiliary vertices (green/light vertices) in
$\mathcal{G}_{\mathcal{Q}_c}^{\mathcal{C}_{\text{out}}\{enc\}}$.  The
corresponding encoding circuit is given by the middle circuit in
FIG.~\ref{fig:C422}, where ``$\slash$" on each line indicates that
there is a set of qubits, not just one.  For instance, the line
corresponding to $\ket{q_0}$ represents the $4$ input qubits ($4$
white vertices in the left graph of FIG.~\ref{fig:C422} ), the line
corresponding to $\ket{q_1}$ represents the $8$ auxiliary qubits, and
the line corresponding to $\ket{q_2}$ represents the $4$ output qubits
of a single inner code $\mathcal{Q}_{\text{in}}$. The graph
$\mathcal{G}_c^{\mathcal{C}_c}$ of the concatenated code
$\mathcal{Q}_c$ can be obtained from the encoding graph
$\mathcal{G}_{\mathcal{Q}_c}^{\mathcal{C}_{\text{out}}\{enc\}}$ by
applying Corollary~\ref{cor:mainsimple}. The result is shown as the
right graph in FIG.~\ref{fig:C422}.  The blue/dashed lines are the
edges obtained by generalized local complementation. 

\begin{figure}[htp]
\centerline{
\includegraphics[width=1.5in]{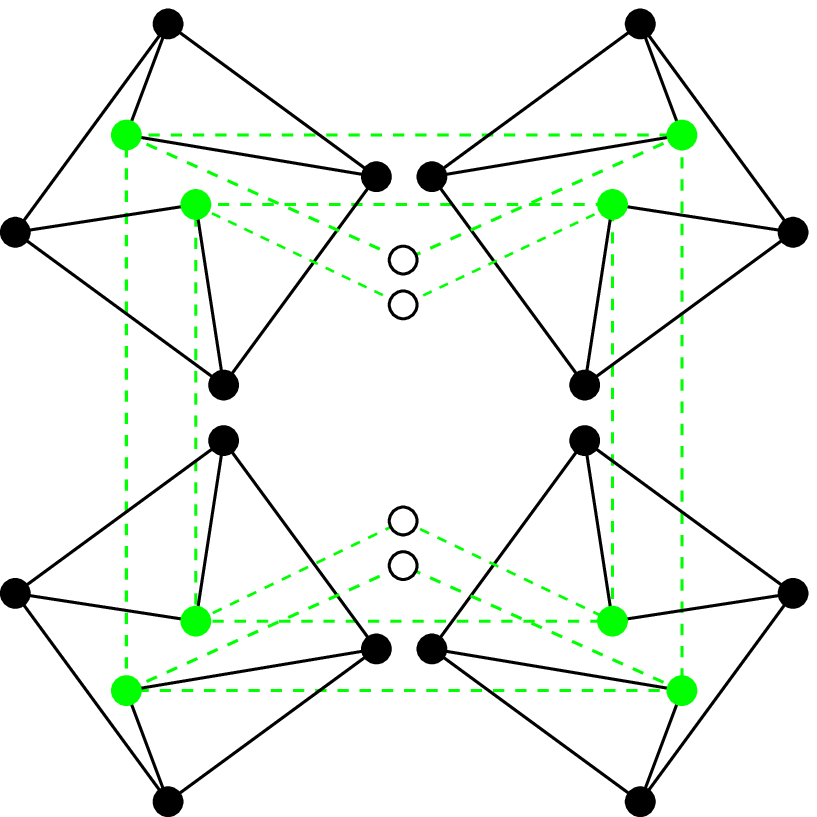}
\quad
\includegraphics[width=2.5in]{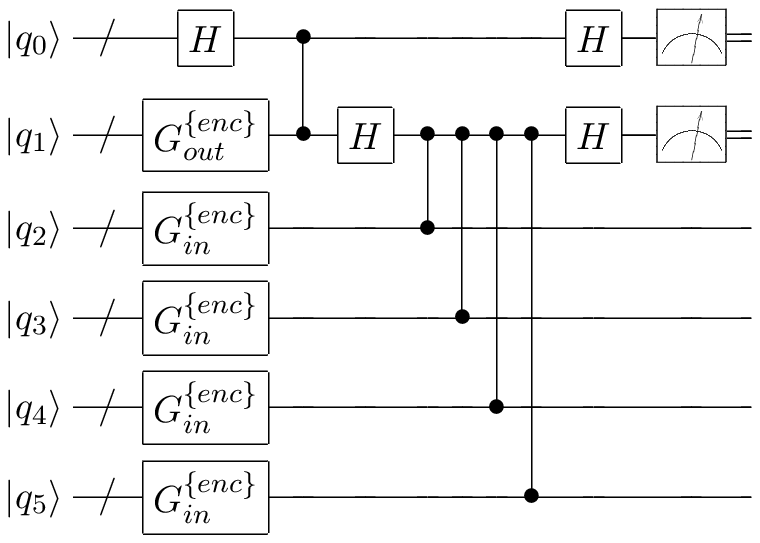}
\quad
\includegraphics[width=1.5in]{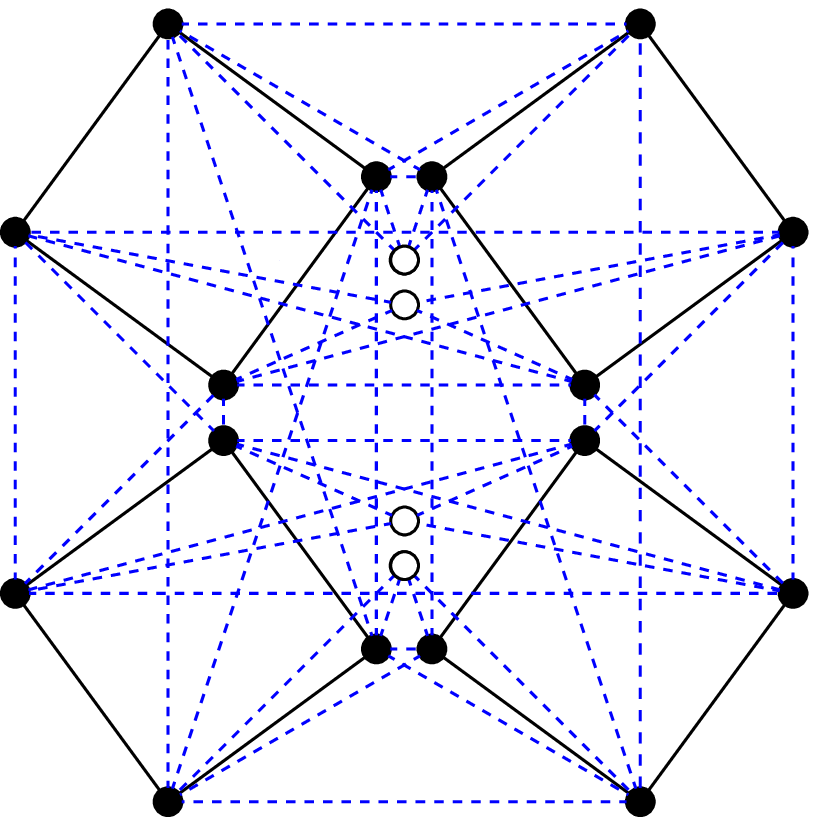}
}
\caption{Graphs and encoding circuit for the concatenated
  $[[16,4,4]]_2$ code obtained by concatenating an inner code
  $[[4,2,2]]_2$ with an outer code $[[4,2,2]]_{2^2}$.}
\label{fig:C422}
\end{figure}

\section{Generalized concatenated codes}\label{sec:Sect_VI}

In this section, we discuss the application of our main result to the
case of generalized concatenated quantum codes (GCQCs). The
construction of GCQCs has been recently introduced in
\cite{GSSSZ,GSZ}.  It resulted in many new QECCS, both stabilizer
codes and nonadditive codes.

A GCQC is derived from an inner quantum code
$\mathcal{Q}_{\text{in}}^{(0)}=((n,q_1q_2\cdots q_r,d_1))_p$, which is
first partitioned into $q_1$ mutually orthogonal subcodes
$\mathcal{Q}_{\text{in}\{{i_1}\}}^{(1)}$ ($0\leq i_1\leq q_1-1$), where each
$\mathcal{Q}_{\text{in}\{{i_1}\}}^{(1)}$ is an $((n,q_2\cdots q_r,d_2))_p$
code. Then each $\mathcal{Q}_{\text{in}\{{i_1}\}}^{(1)}$ is partitioned into
$q_2$ mutually orthogonal subcodes
$\mathcal{Q}_{\text{in}\{{i_1i_2}\}}^{(2)}$ ($0\leq i_2\leq q_2-1$), where
each $\mathcal{Q}_{\text{in}\{{i_1i_2}\}}^{(2)}$ has parameters
$((n,q_3\cdots q_r,d_3))_p$, and so on.  Finally, each
$\mathcal{Q}_{\text{in}\{{i_1i_2\ldots i_{r-2}}\}}^{(r-2)}$ is partitioned
into $q_{r-1}$ mutually orthogonal subcodes
$\mathcal{Q}_{\text{in}\{i_1i_2\ldots i_{r-1}\}}^{(r-1)}=((n,q_r,d_r))_p$ for
$0\leq i_{r-1} \leq q_{r-1}-1$.  Thus
\begin{alignat}{7}
\mathcal{Q}_{\text{in}}^{(0)}=\bigoplus_{i_1=0}^{q_1-1}\mathcal{Q}_{\text{in}\{i_1\}}^{(1)},&\quad &
\mathcal{Q}_{\text{in}\{i_1\}}^{(1)}=\bigoplus_{i_2=0}^{q_2-1}\mathcal{Q}_{\text{in}\{i_1i_2\}}^{(2)},&\quad &
\ldots,\label{eq:inner_dec}
\end{alignat}
and $d_1\le d_2\le\ldots\le d_r$.  In addition, we take as outer codes
a collection of $r$ quantum codes
$\mathcal{Q}_{\text{out}}^{(1)},\ldots,\mathcal{Q}_{\text{out}}^{(r)}$,
where $\mathcal{Q}_{\text{out}}^{(j)}$ is an $((n',K'_j,d'_j))_{q_j}$
code over the Hilbert space $\mathcal{H}_{q_j}^{\otimes n'}$.

The generalized concatenated code $\mathcal{Q}_{gc}$ is a quantum code
in the Hilbert space $\mathcal{H}_{q}^{\otimes nn'}$ of dimension
$K'=K'_1K'_2\cdots K'_r$. The detailed construction of
$\mathcal{Q}_{gc}$ can be found in \cite{GSZ}. Here we only emphasize
that the essence of the ``generalization", which is different from the
usual concatenated quantum codes, is that the outer code is actually a
product of $r$ outer codes, and the inner code is nest-decomposed to
specify how those product of outer codes are encoded into each inner
code. Therefore, similar to a concatenated code $\mathcal{Q}_c$, a
GCQC $\mathcal{Q}_{gc}$ with a graph inner code
\begin{equation}
\mathcal{Q}_{\text{in}}^{(0)}=(\mathcal{G}_{\text{in}}^{(0)},\mathcal{C}_{\text{in}}^{(0)})
\end{equation}
and $r$ CWS outer codes
\begin{equation}
\mathcal{Q}_{\text{out}}^{(j)}=(\mathcal{G}_{\text{out}}^{(j)},\mathcal{C}_{\text{out}}^{(j)})
\end{equation}
naturally has an encoding graph,
denoted by $\mathcal{G}_{\mathcal{Q}_{gc}}^{\{enc\}}$, and the corresponding encoding circuit
is given by the following procedure.

\begin{procedure} (Encoding circuit for generalized
concatenated code $\mathcal{Q}_{gc}$)
\begin{enumerate}
\item Apply the encoding circuits of $\mathcal{Q}_{\text{out}}^{(j)}$
  that encodes $K'_j$ states (or $k'_j\log_p{q_j}$ qupits if
  $\mathcal{C}_{\text{out}}^{(j)}$ is linear) into $n'\log_p{q_j}$
  qupits which we call $q_1, \dots ,q_{n'\log_p{q_j}}$.
\item Apply $n'$ copies of the circuit that gives the graph state
corresponding to $\mathcal{G}_{\text{in}}$.
\item For each $j=1,\ldots,r$, apply $H^\dagger$ on all qupits $q_1, \dots ,q_{n'\log_p{q_j}}$.
\item Apply the corresponding controlled-$Z$ operators between these qupits and
the graph states of $\mathcal{G}_{\text{in}}$.
\item For each $j=1,\ldots,r$, apply $H$ on $q_1, \dots ,q_{n'\log_p{q_j}}$.
\item For each $j=1,\ldots,r$, measure $q_1, \dots ,q_{n'\log_p{q_j}}$
  in the computational basis.
\end{enumerate}
\end{procedure}

Notice that Steps 3, 4, 5 remain the same as those given in Procedure
\ref{pro:enccon}.  Consequently, a similar result as in Corollary
\ref{cor:mainsimple} holds for constructing GCQCs as well.
\begin{corollary}
\label{cor:gc}
$\mathcal{Q}_{gc}=(\mathcal{G}_{gc},\mathcal{C}_{gc})$,
where $\mathcal{G}_{gc}$ can be obtained for the encoding graph
$\mathcal{G}_{\mathcal{Q}_{gc}}^{\{enc\}}$ via Procedure~\ref{pro:main}
and $\mathcal{C}_{gc}$ is the classical generalized concatenated code
with inner code $\mathcal{C}_{\text{in}}^{(0)}$ (with corresponding decomposition
given by the decomposition of $\mathcal{Q}_{\text{in}}^{(0)}$,
see \cite{GSZ} for details) and the outer codes $\mathcal{C}_{\text{out}}^{(j)}$
($j=1,\ldots,r$).
\end{corollary}

For an example, the left graph of FIG.~\ref{fig:GC422} is the encoding
graph
$\mathcal{G}_{\mathcal{Q}_{gc}}^{\mathcal{C}_{\text{out}}^{(0)}\{enc\}}$
of the GCQC $\mathcal{Q}_{gc}$ with a $[[4,2,2]]_2$ inner code code
that is decomposed into two copies of a code $[[4,1,2]]_2$.  There are
two different outer codes $[[4,4,1]]_2$ and $[[4,2,2]]_2$.  Note that
there are $4+2=6$ input vertices (white vertices) and $8$ auxiliary
vertices (green/light vertices).  The corresponding encoding circuit
is given by the middle circuit in FIG.~\ref{fig:GC422}, where
``$\slash$" on each line means that the line actually represents a set
of qubits.  For instance, the line corresponding to $\ket{q_{00}}$
represents the $4$ input qubits of the $[[4,4,1]]_2$ outer code, the
line corresponding to $\ket{q_{01}}$ represents the $2$ input qubits
of the $[[4,2,2]]_2$ outer code, the lines corresponding to
$\ket{q_{10}}$ and $\ket{q_{11}}$ represents the $4$ auxiliary qubits
of the $[[4,4,1]]_2$ and the $[[4,2,2]]_2$ outer codes, respectively,
and the line corresponding to $\ket{q_2}$ represents the $4$ output
vertices in a single $\mathcal{Q}_{\text{in}}$. To obtain the graph
$\mathcal{G}_{gc}^{\mathcal{C}_{gc}}$ of the concatenated code
$\mathcal{Q}_{gc}$ from the encoding graph
$\mathcal{G}_{\mathcal{Q}_{gc}}^{\mathcal{C}_{\text{out}}^{(0)}\{enc\}}$
apply Corollary~\ref{cor:gc}. The result is shown as the right graph
in FIG.~\ref{fig:C422}.

\begin{figure}[htp]
\centerline{
\includegraphics[width=1.5in]{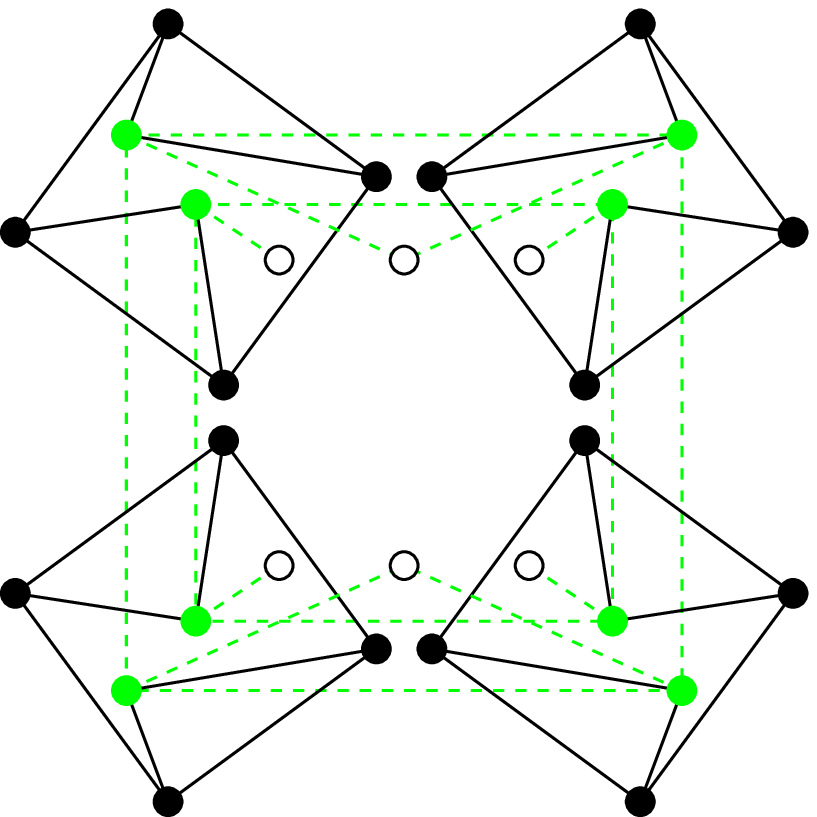}
\quad
\includegraphics[width=2.5in]{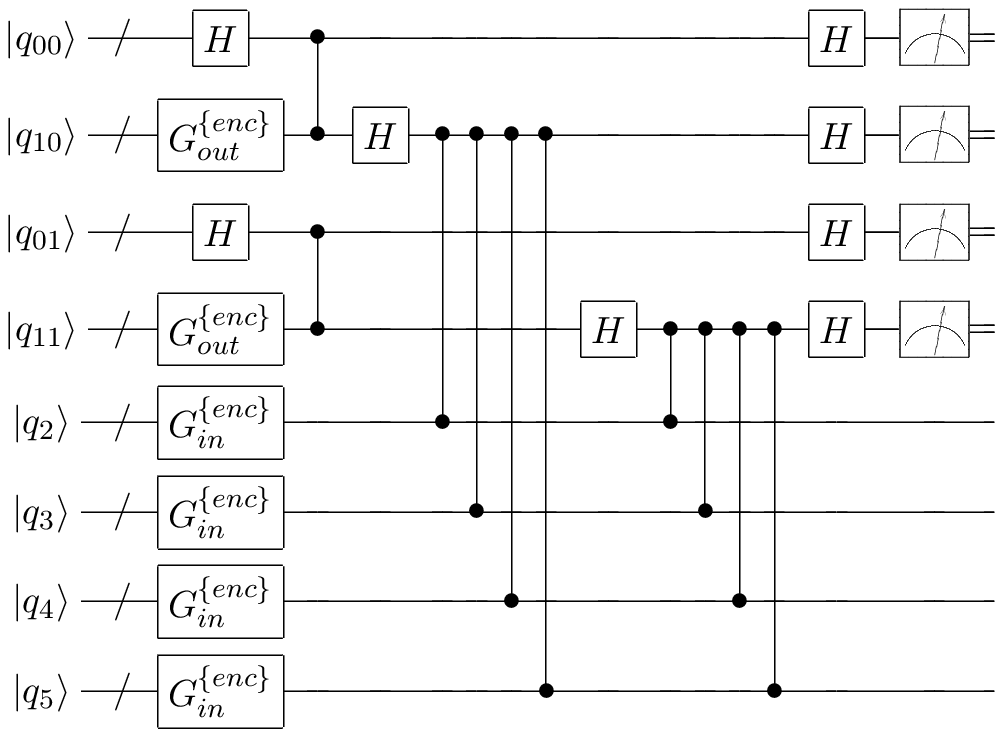}
\quad
\includegraphics[width=1.5in]{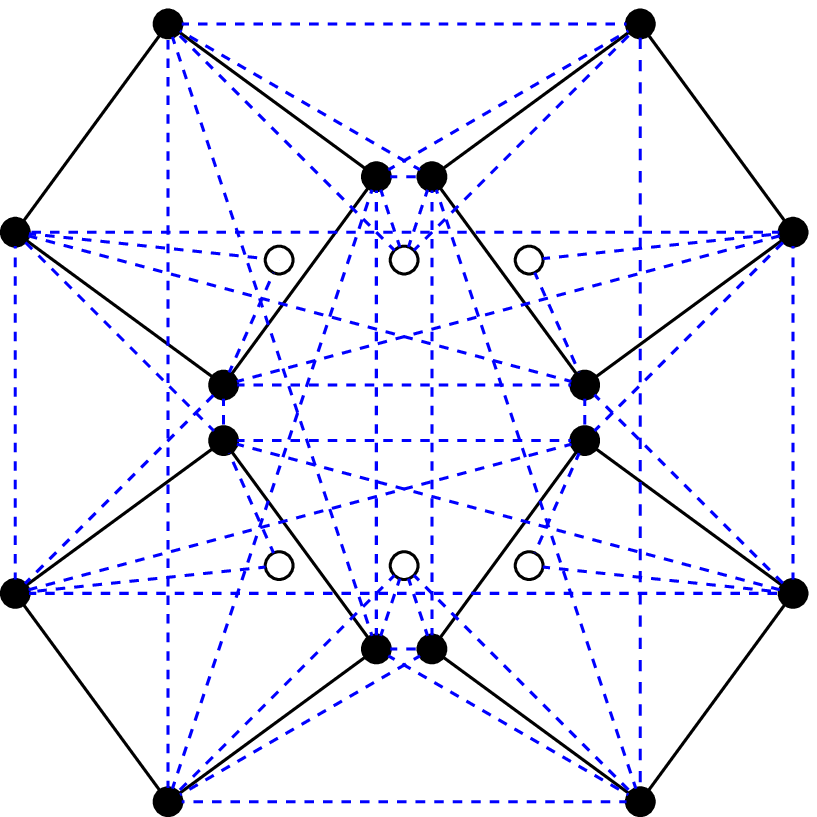}
}
\caption{Graphs and encoding circuit for the generalized concatenated
  $[[16,6,2]]_2$ code, derived from an inner code $[[4,2,2]]_2$ and
  outer codes $[[4,2,2]]_2$ and $[[4,4,1]]_2$.}
\label{fig:GC422}
\end{figure}

\section{Conclusion and discussion}\label{sec:Sect_VII}

In this paper we develop a systematic method for constructing
concatenated quantum codes based on ``graph concatenation", where
graphs representing the inner and outer codes are concatenated via a
simple graph operation called ``generalized local complementation."
The outer code is chosen from a large class of quantum codes, called
CWS codes, which includes all the stabilizer codes as well as many
good nonadditive codes. The inner code is chosen to be a stabilizer
code. Despite the restriction that the inner code must be a stabilizer
code, our result applies to very general situations---both binary and
nonbinary concatenated quantum codes, and their generalizations.

Our results indicate that graphs indeed capture the ``quantum part" of
the QECCs. Once the graph part is taken care of, the construction of
quantum code is reduced to a pure classical problem.  This was
essentially the idea of the CWS framework (i.e., the problem of
constructing a CWS quantum code is reduced to the problem of finding a
classical code with error patterns induced by a given graph).  Here we
have demonstrated that this idea extends to the construction of
(generalized) concatenated quantum codes as well (i.e., to construct
(generalized) concatenated quantum codes, given the rule of graph
concatenation, one only needs to construct the (generalized) classical
concatenated codes).  We believe that our results shed light on the
further understanding of the role that graphs play in the field of
quantum error correction and other related areas in quantum
information theory.

\textit{Acknowledgment} 
We thank Runyao Duan for helpful discussions. BZ is supported by NSERC
and QuantumWorks.  Centre for Quantum Technologies is a Research
Centre of Excellence funded by Ministry of Education and National
Research Foundation of Singapore.


\end{document}